\documentclass[reqno,11pt]{amsart}

\usepackage{amsmath}
\usepackage{amsfonts}
\usepackage{amssymb}
\usepackage{amsbsy}
\usepackage{amsthm}

\usepackage{latexsym}
\usepackage{caption,subcaption,color,booktabs}
\usepackage{graphicx}

\usepackage{fourier}

\usepackage[switch*]{lineno}

\usepackage{setspace}

\newtheorem{theorem}{Theorem}

\newtheorem{definition}[theorem]{Definition}

\newtheorem{proposition}[theorem]{Proposition}

\newtheorem{remark}[theorem]{Remark}

\newcommand{\norm}[1]{\left\Vert#1\right\Vert}
\newcommand{\inpd}[2]{\left\langle #1, #2 \right\rangle}

\newcommand{\abs}[1]{\left\vert#1\right\vert}

\newcommand{\Real}{\mathbb R}

\newcommand{\eps}{\varepsilon}

\newcommand{\Ker}{\mbox{Ker}}
\newcommand{\img}{\mbox{Img}}
\DeclareMathOperator{\spn}{span}

\renewcommand{\vec}[1]{\boldsymbol{#1}}

\newcommand{\mfd}{\mathcal{M}}
\newcommand{\tang}{{T}}

\newcommand{\bX}{\vec{X}}
\newcommand{\bx}{\vec{x}}
\newcommand{\by}{\vec{y}}
\newcommand{\bW}{\vec{W}}
\newcommand{\ud}{\vec{u}}
\newcommand{\nd}{\vec{n}}
\newcommand{\bd}{\vec{b}}
\newcommand{\phid}{\vec{\phi}}
\newcommand{\varphid}{\vec{\varphi}}
\newcommand{\pd}{\vec{p}}

\newcommand{\transpose}{\textsf{T}} 
\newcommand{\T}{\ensuremath{{\rm T}}}

\renewcommand{\d}{\ensuremath{\,\mathrm{d}}}

\newcommand{\eye}{{I}}

\newcommand{\FW}{{ Freidlin-Wentzell} }

\newcommand{\leqref}[1]{Eq.~\eqref{#1}}

\usepackage{lipsum}
\makeatletter
\g@addto@macro{\endabstract}{\@setabstract}
\newcommand{\authorfootnotes}{\renewcommand\thefootnote{\@fnsymbol\c@footnote}}%
\makeatother
\begin{document}

%
%


\title{ }

\begin{center}
  \LARGE 
Finding Transition Pathways on Manifolds \par \bigskip

  \normalsize
  \authorfootnotes
  Tiejun Li \footnote{email: {\it tieli@pku.edu.cn}. Mailing address: School of Mathematical Sciences, Peking University, Beijing 100871.}\textsuperscript{1},  Xiaoguang Li\footnote{email: {\it lxg1023@pku.edu.cn}. Mailing address: School of Mathematical Sciences, Peking University, Beijing 100871.}\textsuperscript{1},
  Xiang Zhou\footnote{email: {\it xizhou@cityu.edu.hk}. Mailing address: 
  Department of mathematics, City University of Hong Kong, Tat Chee Ave, Kowloon, Hong Kong.}\textsuperscript{2} \par \bigskip

  \textsuperscript{1}LMAM and School of Mathematical Sciences,  Peking University,  China \par
  \textsuperscript{2}Department of Mathematics,
        City  University of Hong Kong, Hong Kong\par \bigskip

  \today
\end{center}



\date{}

\maketitle

\section*{Abstract}
We consider  noise-induced transition paths in  randomly perturbed dynamical systems  on  a smooth  manifold.   The classical \FW large deviation theory in Euclidean spaces is  generalized  and 
new forms of action functionals  are derived in the spaces of functions and the space of curves to accommodate the intrinsic constraints  associated with the manifold. 
Numerical methods are proposed to compute  the minimum action paths for the systems  with constraints.
The examples of conformational transition paths for a single and double rod molecules arising in polymer science are numerically investigated. 


\section{Introduction}

A large number of interesting behaviours of stochastically perturbed dynamical systems are
closely related to rare but important transition events between metastable
states. Such rare events play a major role in chemical reactions,
conformational changes of biomolecules,  nucleation events and the
like. 
Theoretical understanding of  such transition events and transition
paths has attracted a lot of attentions for many years \cite{Kramers, FW1998}. 
The   model under consideration  is the following (It\^{o}) stochastic differential equations (SDEs) in $\Real^n$ with  small noise amplitude 
\begin{equation}
\d X_t  = b(X_t)  \d t + \sqrt{\eps} \sigma(X_t)  \d W_t.
\end{equation}
The drift term $b(x)$  could be the gradient form of a potential energy function or  have a  rather  general  form.  
The diffusion matrix $\sigma(x)$ is assumed  uniformly  non-degenerate.   $b$ and $\sigma$ satisfy  the  regular smoothness conditions
that  are  Lipschitz continuous and bounded.

According to the \FW large deviation theory \cite{FW1998}, in the asymptotics of vanishing noise $\eps\downarrow 0$,
the most probable transition pathway  is the minimizer of   the \FW action functional $S$, being
\begin{equation}\label{eqn:S}
S_T[\varphi] =  \frac12  \int_0^T \|{\sigma^{-1}(\varphi) (\dot{\varphi}-b(\varphi)})\|^2_2 \d t.
\end{equation}
Based on this principle of least action, a few numerical methods,  such as the Minimum Action Method (MAM) and its adaptive version \cite{weinan-MAM2004,aMAM2008}, have been proposed and developed for a fixed time interval $[0, T]$ of interest. 
Another different  formulation of the \FW theory,  based on  Maupertuis' principle \cite{Landau-Lifshitz-Mechanics}, is the geometric Minimum Action Method (gMAM) on the space of curves \cite{Heymann2006}.   The path given by the gMAM 
can be viewed as the minimum action path of the original \FW action  for an optimal $T$.
In the special case that   $b(x)=-\nabla V(x)$ and $\sigma(x)\equiv 1$,  the minimum action path  is minimum energy path
and the string method \cite{String2002} is applicable to  identify  this path. 

In practical applications, the dynamics may be subject to one or more constraints, 
such as the constant length of rigid   molecules, the conservation of mass, etc. 
These constraints limit the system to live in a particular manifold $\mfd$, decided by all the constraints, rather than in Euclidean space $\Real^n$.
Even when the stochastic perturbation is applied, the resulting stochastic system 
still has to satisfy these physical  constraints.  Thus, one needs to  model the {\it stochastic} system as 
an SDE  on a manifold  $\mfd$ rather than a deterministic flow $\dot{x}=b(x)$ on $\mfd$ perturbed by  the noise  freely in $\Real^n$. The problem in the latter case is that the perturbed stochastic system does not conserve the quantities associated with $\mfd$.  This subtlety actually implies a different form of the resulting action functional,
although in the special case of isotropic noise, i.e., $\sigma(x)\equiv1$, the action functionals from the two formulations are equivalent.
In this reversible case, the  modification  of the original string method by directly applying the constraints  to the path 
is applicable  \cite{ConstString}. It is also justified below in our paper, that 
the straightforward use of \eqref{eqn:S} by solving the constraint minimization problem  $\inf_{\varphi\in \mfd} S_T[\varphi]$
is applicable when the diffusion tensor $\sigma$ is isotropic.

We present here  a rigorous derivation  of action functionals  for general (non-degenerate) diffusion tensor $\sigma(x)$ 
by  starting from the SDE  on the manifold  $\mfd$ with vanishing noise.  Although  an abstract analogy  of  \FW action functional
can be readily accessed  (Section \ref{SDEmfd}),  the practical application calls for the expressions when the underlying  SDE 
is written on $\mfd$ as   embedded in the Euclidean  space $\Real^n$.
This setting up in particular  caters  for the case under study:  $\mfd=\{x\in \Real^n:  c_i(x)=0\}$ where $c_i$ are constraints of the system,
where a projection operator $\Pi$ from $\Real^n$ to $\tang \mfd$,
the tangent space of $\mfd$,  can be introduced.
By handling the degeneracy issue of the projected diffusion noise, we derive the  new forms of  action functionals. 
Our analysis  suggests that 
the resulting forms of action functional on the manifold may differ from  the traditional one like \leqref{eqn:S}.
The difference comes from the discrepancy of the metric:  
The diffusion induces a metric  $a(x)=\sigma(x)\sigma(x)^\transpose$, and the 
minimum action path on $\mfd$ could be viewed as a geodesic  (at least for the pure diffusion case when $b=0$) on $\mfd$ equipped with the metric $a$,
but  the projection $\Pi$  uses the $L^2$-norm of  $\Real^n$.

The paper is organized as follows.  We first discuss the stochastic differential equation 
in the Stratonovich sense on the manifold and the abstract form of the \FW action functional on the manifold in Section 2.
In Section 3,  we  consider the manifold embedded in an Euclidean space and introduce the local projection operator 
to represent  the SDE with coordinates formulation.  The formula of the action functional are then derived on
the space of  functions of time  and the space of  geometric curves.  Section 4 is devoted to the 
numerical methods --- the constrained minimum action method.  The applications to
liquid crystal models are presented in Section 5, where we study the conformational transitions for rod molecules  on $\mathbb{S}^2$ (unit sphere) and $\mathbb{S}^2\times \mathbb{S}^2$.  In Section 6, we present some outlook for other types of transition paths with constraints beyond our current approach. Finally we make the summary.

\section{SDE and large deviation principle  on manifolds}
\label{SDEmfd}

The SDE on the manifold is most conveniently written in  the Stratonovich sense  \cite{Book:Hsu}.
We consider 
\begin{equation}\label{eqn:SDE-Manifold}
\d X_{t}=b(X_{t})\d t +\sqrt{\eps}\sum_{k=1}^{L}\sigma_{k}(X_{t}) \circ \d W_t^k
\end{equation}
on a compact differentiable $d$-dimensional manifold $\mfd$ without boundary. Here $X\in \mfd$, the drift and diffusion functions $b,\{\sigma_{k}\}_{k=1}^{L}\in \tang\mfd$, the tangent bundle of $\mfd$, and $\{W^{k}\}_{k=1}^{L}$ are independent Wiener processes on $\Real$. 
We assume the non-degeneracy condition for diffusion, i.e. 
$$\dim \mbox{span}\{\sigma_{k}(x)\}_{k=1}^{L} = d$$
for any $x\in \mfd$.

From the large deviation theory on manifolds,
under certain regularity conditions on $b$ and $\sigma_k$,  we have the rate functional (or action functional) $S[\varphi]$ as $\eps$ goes to $0$
\begin{equation}\label{eqn:Smfd}
S_{T}[\phi] =   \int_0^T L(\phi,\dot\phi)\d t,
\end{equation}
where $\phi\in C([0,T];\mfd)$ is absolutely continuous, and $\dot\phi$ means the derivative with respect to the time $t$. The Lagrangian $L$ is given by the Legendre transformation of the Hamiltonian $H$ as
\begin{equation}\label{eqn:LagHamil-Legendre}
L(x,y) = \sup_{p\in T^*\mfd} \left \{ \inpd{p}{y} -H(x,p) \right \},
\end{equation}
where $T^*\mfd$ is the cotangent bundle of $\mfd$, $y\in T\mfd$, and $\inpd{\cdot}{\cdot}$ is the dual product between the spaces $T^*\mfd$ and $T\mfd$.

For our system \eqref{eqn:SDE-Manifold} the Hamiltonian $H$ has the form
\begin{equation*}
H(x,p) = \inpd{p}{b(x)}  +  \frac12 \sum_{k=1}^{L} \left \vert \inpd{p}{\sigma_{k}(x)} \right \vert^{2}.
\end{equation*}
Since $H$ is just a quadratic form of $p$, we have the equation for the critical point of \eqref{eqn:LagHamil-Legendre}
\begin{equation}\label{eqn:CriticalPoint}
\frac{\partial H}{\partial p} = b(x)+ \sum_{k=1}^{L} \inpd{p}{\sigma_{k}(x)} \sigma_{k}(x) = y. 
\end{equation}
Note that the type $(0,2)$ covariant symmetric tensor field
$$
\sum_{k=1}^{L} \sigma_{k}\otimes \sigma_{k}\in T^{2}_{0}\mfd
$$
can also be viewed as a mapping
$$
a := \sum_{k=1}^{L} \sigma_{k}\otimes \sigma_{k}: T^*\mfd\longrightarrow T\mfd
$$
by fixing its first or second argument \cite{Book:Bishop}. From the non-degeneracy condition, we have that the mapping $a$ is bijective, thus its inverse $a^{-1}: T\mfd\rightarrow T^{*}\mfd$ is well-defined. Solving the equation \eqref{eqn:CriticalPoint} we get $p = a^{-1}(y-b(x))$ and thus
\begin{equation}\label{eqn:Labs}
L(x,y) = \frac12  \inpd{a^{-1}(y-b(x))}{y-b(x)}.
\end{equation}

To end this section, we comment on  the equivalence of the action functionals  in the large deviation theory
for Stratonovich SDEs and It\^{o} SDE. The equivalent  It\^{o} SDEs corresponding to \leqref{eqn:SDE-Manifold} has an additional term  with order $ O(\eps)$ besides the original $b$.
However, this additional $O(\eps)$ term  uniformly vanishes as $\eps\downarrow 0$ if $\sigma$ and its derivative are bounded.
Thus, by this fact, as shown in  \cite{FW1998},  this    SDE \eqref{eqn:SDE-Manifold}  shares the same large deviation result with 
the one written in the  It\^{o}  sense.

\section{Action Functional}
\label{sec:act}

The   abstract formulation Eqs.  \eqref{eqn:Smfd} and \eqref{eqn:Labs} in Section 2
of the action functional for SDE on the manifold  is applicable for many realistic problems. 
To be more tractable, we consider the manifold $\mfd$ embedded in the Euclidean space $\Real^n$, then  the SDE on $\mfd$ can be treated as  an SDE in $\Real^n$ and the standard \FW action functional can be explicitly calculated.
To represent the aforementioned  embedding, we need  introduce the projection operator for $\mfd$.
 

\subsection{Projection and its inverse}

Assume that $\mfd$ is  embedded in the Euclidean space $\Real^n$ ($n\ge d$). We  introduce $\Pi_x: \Real^{n}\rightarrow T_{x}\mfd\subset \Real^{n}$,  the  orthogonal  projection operator at point $x$  on the considered $d$  dimensional manifold $\mfd$. 
Given  a vector field  $b:\Real^{n}\rightarrow\Real^{n}$ and a uniformly nondegenerated diffusion tensor $\sigma:\Real^{n}\rightarrow \Real^{n\times m}$ $(m\ge n)$,
 the process of interest on $\mfd$ is of the following projection form
 \begin{equation}
\label{eqn:sde-mfd}
dX = \Pi  \Big ( b(X)\d t +   \sqrt{\eps}  \sigma (X) \circ \d W \Big),
\end{equation}
where $\Pi=\Pi_X$.  The subindex of the projection $\Pi$ is sometimes dropped out henceforth. For each sample trajectory in the probability space, if the initial condition $X_0=x\in \mfd$, the solution of \leqref{eqn:sde-mfd} $X_t$ is always on the manifold $\mfd$ for any time $t>0$.

The Hamiltonian for the SDE \eqref{eqn:sde-mfd} \cite{FW1998, Heymann2006} is 
\begin{equation*}
H(x,p) = \inpd{\Pi b(x)}{p} + \frac12 \left\|{\sigma(x) ^\transpose } \Pi p\right\|^2, ~~~ \forall x,p\in\Real^n,
\end{equation*}
where $\inpd{\cdot}{\cdot}$ and $\|\cdot\|$ are the inner product and $L^2$ norm of  $\Real^n$, respectively. We note that $\Pi^\transpose = \Pi$ since it is an orthogonal projection. The corresponding Lagrangian  is defined by the Legendre transformation  as follows
\begin{equation}
\label{eqn:Lag}
\begin{split}
L(x, y)& = \sup_{p\in \Real^n}  \Big( \inpd{y}{p} - H(x,p)  \Big)\\
&=  \sup_{p\in \Real^n} \left ( \inpd{y-\Pi b(x)}{p}  - \frac12 \|{\sigma(x) ^\transpose } \Pi p\|^2\right).
\end{split}
\end{equation}
For any $x$,   $\Pi: \Real^n\to \tang_x \mfd$ is the projection at this point $x$.
Then, each vector $p\in \Real^n$ can be written as $p=p_1+p_2$ where $p_1$ is in the image space, 
$ \img(\Pi)  $ and $p_2$ is in the kernel space, $\Ker(\Pi)$.
Since  $p_1=\Pi p$ and $\inpd{p_2}{\Pi b}=0$, then
\begin{equation*}
\begin{split}
L(x, y)&=  \sup_{p_1\in\img(\Pi), p_2\in \Ker(\Pi)} \left ( \inpd{y}{p_2} + \inpd{y-\Pi b(x)}{p_1} - \frac12 \|{\sigma(x) ^\transpose }  p_1\|^2\right).
\end{split}
\end{equation*}
If $y\notin \img(\Pi)$, then the    $\inpd{y}{p_2}$ term  on the right hands side of the preceding equation  can grow infinitely large and in this case,  $L=+\infty$. Hence, we only need to consider the case that 
 $y\in \img(\Pi)$ henceforth, 
 then it follows that \begin{equation*}
\begin{split}
L(x, y)&=  \sup_{p_1\in\img(\Pi)} \left (  
\inpd{y-\Pi b(x)}{p_1} - \frac12 \|{\sigma(x) ^\transpose }  p_1\|^2_2\right).
\end{split}
\end{equation*}
To solve the above  constrained convex  optimization problem, we seek its dual solution.
Define the optimization  Lagrangian function $\mathcal{L}: \Real^n\times \Real^n \to \Real$ as
\[
\mathcal{L}{(p,\lambda)} = \inpd{y-\Pi b(x)}{p} - \frac12 \|{\sigma(x) ^\transpose }  p\|^2_2
+ \inpd{\lambda}{ (\eye-\Pi)p}
\]
and the dual function
$g(\lambda)=\sup_{p\in \Real^n} \mathcal{L}(p,\lambda)$ for  $\lambda \in \Real^n$, 
where $I$ is the identity matrix.
The   optimal $p$ in  definition of  $g$ is 
\[p^*=a^{-1} (y-\Pi b +(I-\Pi)\lambda)\] 
where  $a^{-1}$ is the inverse of the positive-definite matrix 
$a(x)=\sigma(x)\sigma(x)^\transpose$ evaluated at $x$.
So  the dual function is 
\[g(\lambda)=\mathcal{L}(p^*,\lambda) =\frac12 
\left \|    y-\Pi b(x)  + (I-\Pi)\lambda \right \|^2_a .\]
Here the $a$-weighted norm $\|u\|_a$ associated with the positive definite matrix $a(x)$ 
is $\sqrt{u^\transpose a^{-1} u}$  evaluated at $x$.
Likewise, the $a$-weighted inner product is defined by $\inpd{u}{v}_a:=u^\transpose a^{-1}v$
for $u,v\in\Real^n$.

It is clear that for the quadratic optimization problem, the strong duality holds.
So, we have 
\[ L(x,y)=\inf_{\lambda} \ g(\lambda) =\inf_{\lambda} \ \frac12  \left\|   y-\Pi b(x)  + (I-\Pi)\lambda \right\|^2_a.\]
By introducing \[u= y-\Pi b(x)+ (I-\Pi)\lambda,\]
the infimum  of $g$ becomes
\[\inf_\lambda g(\lambda)=\min_{u\in \Real^n, \Pi u = y-\Pi b(x)}\frac12 \|u\|^2_a.\]
Therefore for $y \in \img(\Pi)$,
\begin{equation}
\label{eqn:L-1}
L(x,y) = \min_{u\in \Real^n, \Pi u = y-\Pi b(x)} \frac12 \|u\|_a^2.
\end{equation}
By the duality theory, we also have that \begin{equation}
\label{eqn:gm}
\frac{\partial L}{\partial y}=p^*=a^{-1}u^*
\end{equation}
where $u^*$
is the solution of the minimization problem \eqref{eqn:L-1}.

From the constraint   $\Pi u = y-\Pi b(x)$ for the minimization problem \eqref{eqn:L-1}, one may formally view  
$u$ as an element in the set $\Pi^{-1} (y-\Pi b(x))$  which has 
the minimal  $a$-weighted  norm.
To ease the presentation, we redefine $\Pi^{-1}$ as follows.
\begin{definition}
\label{def:inverse}
Let $\Pi$ be an orthogonal projection matrix in $\Real^n$
and $a$ be  a positive definite matrix. 
For any $ v \in \img(\Pi)$,  
we define 
$\Pi^{-1} v$ as the vector $u^*\in \Real^n$ such that $u^*$
solves 
$ \underset{\Pi u= v}\min \|u\|_a$.
\end{definition}

The above defined $\Pi^{-1} v$ for a given $v\in \img (\Pi)$ is unique since $a$ is not singular. 
We point out that $\Pi^{-1}$ is not exactly an inverse of $\Pi$ in strict sense because although
$\Pi \circ \Pi^{-1} $ is identity restricted on  the space $\img (\Pi)$, it is generally invalid that  $\Pi^{-1}  (\Pi v)=v$.
This generalized inverse $\Pi^{-1}: \img(\Pi) \to \Real^n$  depends on the metric induced by $a$.
If $a$ is a scalar matrix, then $\Pi^{-1}$ is identity restricted on $\img (\Pi)$. 
In the following derivation of the action functional for the SDE \eqref{eqn:sde-mfd},
 $a=\sigma\sigma^T$ is implicitly applied in the context where $\Pi^{-1}$ appears. 
Before our derivation,  we first point out some useful properties of $\Pi^{-1}$.

\begin{proposition}
\label{prop:aorth}
For every  $ w \in \Ker(\Pi)$ and   $v \in \img(\Pi)$,  we have
\begin{equation*}
\inpd{\Pi^{-1}v}{w}_a=0.
\end{equation*}
\end{proposition}
\begin{proof}
Let $u^*=\Pi^{-1}v$. Define   $u_\theta=u^*+ \theta w$, $\forall \theta\in\Real$.
Then $\Pi u_\theta = \Pi u^* = v$ for all $\theta$. So, the function 
$f(\theta)\triangleq \| u_\theta\|^2_a$ has a minimal value  $\|u^*\|^2_a$ at $\theta_{0}=0$.
Consequently, $f'(\theta_0)=0$ and $\inpd{u^*}{w}=0$ follows.
\end{proof}

\begin{proposition}
\label{cor:atri}
For any vector $v\in \img(\Pi)$, it is true that
\begin{align*}
\| \Pi^{-1}v\|_a^2 &=  \inpd{v}{\Pi^{-1}v}_a,  \\
\|v\|^2_a &=\|\Pi^{-1}v\|^2_a+\| {v-\Pi^{-1} v}\|^2_a.  
\end{align*}
\end{proposition}
\begin{proof} The first equality is because $v-\Pi^{-1}v\in \Ker(\Pi)$ and Proposition \ref{prop:aorth}.
Then it follows that 
$\|v\|^2_a  = \inpd{v}{\Pi^{-1}v}+ \inpd{v}{v-\Pi^{-1} v}_a=\|\Pi^{-1}v\|^2_a+ \inpd{v}{v-\Pi^{-1} v}_a$.
Since $\inpd{\Pi^{-1}v}{v-\Pi^{-1} v}=0$ due to Proposition \ref{prop:aorth}, 
then $\|v\|^2_a  =  \|\Pi^{-1}v\|^2_a+\|v-\Pi^{-1} v\|^2_a$. 

\end{proof}

\begin{proposition}
\label{cor:xi}
If  $\dim \Ker(\Pi)=K$,
and $ \Ker(\Pi)=
 \spn\{\xi_k: k=1,\cdots,K\}$, then for any $v\in\img(\Pi)$, 
\[  \|v \|_a^2 = \|\Pi^{-1}v\|_a^2 + \|\hat{v}\|_M^2 ,\]
where $M=(M_{ij})=\inpd{\xi_i}{\xi_j}_a$, $i,j=1,2,\ldots, K$
and $\hat{v}= (\hat{v}_k)=\inpd{v}{\xi_k}_a$, $ k=1,\ldots,K$.
\end{proposition}
\begin{proof}
Write $v-\Pi^{-1}v=\sum_{k}\lambda_k \xi_k$, then
these $\lambda_k$'s minimize $\|v-\sum_{k}\lambda_k \xi_k\|^2_a$.
So, $\lambda=(\lambda_1,\lambda_2,\ldots,\lambda_K)^\transpose$
satisfy $M\lambda = \hat{v}$.
Note that $\|v-\Pi^{-1}v\|_a^2=\|\sum_{k}\lambda_k \xi_k\|_a^2=\lambda^\transpose M \lambda=
\hat{v}^\transpose M^{-1}\hat{v}=\|\hat{v}\|_M^2$.
The conclusion is then immediate  from Proposition \ref{cor:atri}.
\end{proof}

\subsection{Freidlin-Wentzell action functional on $\mfd$}  
\label{sec:FW-mfd}

Given a starting point $A$ and an ending point  $B$ on $\mfd$ as well as a fixed time interval $[0,T]$,
we consider an absolute continuous path  $\phi$ on the manifold $\mfd$   connecting  the two points $\phi(0)=A$
and $\phi(T)=B$. 
By Eqs. \eqref{eqn:Smfd} and \eqref{eqn:L-1},  the action functional (or the rate function) for  
SDE \eqref{eqn:sde-mfd} in the vanishing noise limit is 
\begin{equation}
\label{eqn:rate-mfd}
S_T^\mfd [\phi]  = \begin{cases}
\underset{u}\inf \bigg \{  \frac12 \int_0^T  \|u\|_a^2 \d t:  \   \dot{\phi}-\Pi b(\phi)= \Pi  u, 
\phi(0)=A, \phi(T)=B \bigg \}, ~~\mbox{   if } \dot{\phi}\in \img \Pi_\phi\\
+\infty,  ~~~\mbox{  otherwise}.
\end{cases}
 \end{equation}
Here the projection $\Pi=\Pi_{\phi(t)}$.
Since the admissible path $\phi$  (i.e., $S^{\mfd}_{T}[\phi] < \infty$)
has its tangent $\dot{\phi}$ in the tangent space 
of the manifold $\mfd$,  the entire  admissible path $\phi$ is on $\mfd$. 
The form of the functional 
\leqref{eqn:rate-mfd} can also be formally derived  by  the contraction principle \cite{Varadhan1984,FW1998}.

 In \eqref{eqn:rate-mfd},   $u$  is a function of $t$ and is  equal to 
 $\Pi  ^{-1} ( \dot{\phi}-\Pi b )$ for $t\in [0,T]$ by Definition \ref{def:inverse}.
 Then the action functional \eqref{eqn:rate-mfd}  (for finite value) is  
\begin{equation}\label{S}
S_T^\mfd [\phi]  =   \frac12 \int_0^T   \left \| \Pi  ^{-1} ( \dot{\phi}-\Pi b )  \right \|_a^2 \d t
\end{equation}
  defined over the admissible set 
\begin{equation}
\label{set:A}
\begin{split}
\mathcal{A}=\{\phi \in C([0,T];\mfd):   
\phi(0)=A, \phi(T)=B, 
 ~\phi \mbox{ is absolutely continuous}\},
\end{split}
\end{equation}
which is equivalent to 
\begin{equation}
\label{set:A1}
\begin{split}
\mathcal{A}'=\{\phi \in C([0,T];\Real^n):  &~~\dot{\phi} \in  \img \Pi_\phi, 
\phi(0)=A, \phi(T)=B, \\& ~~\phi \mbox{ is absolutely continuous}\}.
\end{split}
\end{equation}
When the noise is isotropic, i.e., $\sigma(x)$ is a scalar matrix $\sigma I$, then $\|u\|_a = \|u\|/ \sigma^2$. In such a case,  $\Pi^{-1}$ is identity. Then 
the principle of least action is 
\[ \inf_{\phi\in \mathcal{A}}  S^1_T [\phi]\]
where 
\begin{equation}
\label{eqn:S1}
S^1_T [\phi]  \triangleq   
  \frac12 \int_0^T  
\left \| \dot{\phi}- \Pi b(\phi) \right \|_a^2 \d t.
 \end{equation}
 This $S^1_T$ is  the action functional corresponding to  an  SDE  similar to \leqref{eqn:sde-mfd}, but 
without the projection of the random forcing term, i.e.,
\[ \d X = \left(\Pi b(X) \right)\d t + \sqrt\eps \sigma(X) \d W,\]  
whose solution $X_t$ is not  on $\mfd$.  
In general, these two actions are different 
and satisfy 
\[ S^\mfd_T[\phi]\leq S^1_T[\phi]\]
by Proposition \ref{cor:xi}.

\begin{remark}
One naive   approach   might be to solve the  following
minimization problem 
\begin{equation}
\label{eqn:S0}
S^0 [\phi]  =  
\bigg \{  \frac12 \int_0^T  
\| \dot{\phi}- b(\phi) \|_a^2\d t: 
\Pi\dot{\phi} = \dot{\phi}, \ 
\phi(0)=A, \phi(T)=B 
\bigg \}.
 \end{equation}
Note that  $S^0-S^1= \frac12 \int_0^T \| b(\phi)-\Pi b(\phi)\|_a^2dt
$. Even for the gradient system and isotropic  constant diffusion 
$a=\sigma^2 I$, the solutions for $S^0$ and $S^1$ are different.
It is not correct to use  $S^0$ for the constrained SDE problem.
\end{remark}

If we furthermore assume that $b(x)=-\nabla V(x)$ and $\sigma(x)=\eye$, then the system is a gradient system on $\mfd$.
The variational problem $\inf_T \inf_{\phi\in \mathcal{A}}  S^1_T [\phi]$
has a solution which is the minimum energy path on $\mfd$ which satisfies that 
$\dot{\phi} \parallelslant \Pi b(\phi) $ and it follows that the extension of the string method 
works for this case by evolving each image on the string according to the flow $\Pi b(x)$ on $\mfd$ and
applying the  reparametrization  on $\mfd$.
So, our above derivation  justifies the algorithm \cite{ConstString} for this gradient case, but our form \eqref{S}
 is applicable to general cases.

 \subsection{Geometric action functional  on $\mfd$}
 
 The geometric formulation of the action function, developed in \cite{Heymann2006}, does not involve time explicitly
 and allow the variation of the time interval. 
If we consider  the original  formulation of \FW theory as analogy of  Lagrangian mechanics
for the trajectory of a particle, then
the geometric action functional in \cite{Heymann2006} correspond to   the Maupertuis' principle
($\S 44$, \cite{Landau-Lifshitz-Mechanics})
for the curve the particle travels. 

In the next, we consider the geometric action functional $\hat{S}$ for the SDE \eqref{eqn:sde-mfd}
on the manifold $\mfd$.
Suppose that a curve $\gamma$ on $\mfd$ is parametrized  as  $\gamma=\varphi(\alpha)$, with $\alpha\in[0,1]$, for instance, $\alpha$ being the arc length parameter.
Then the 
geometric action (also called {\it abbreviated action}  \cite{Landau-Lifshitz-Mechanics}) is 
the following  line integration along $\gamma$
\begin{equation*}
\hat{S}[\varphi]  = \int_\gamma 
\inpd{p}{ \d \varphi}
\end{equation*}
subject to the constraint $H(\varphi,p)=0$, where $p=\partial L/\partial y (\varphi, \dot{\varphi})$ is the generalized momentum,
$L$ is the Lagrangian defined in \eqref{eqn:Lag} and $\dot{\varphi}$ is the time derivative (velocity). 
By \leqref{eqn:gm}, this generalized momentum is 
 \begin{equation}\label{eqn:Gp}
  p = a^{-1}\Pi^{-1}  (\dot{\varphi}  - \Pi b(\varphi)).
  \end{equation}
So, $\hat{S}$ has the following expression,  
\[ 
\begin{split}
\hat{S}[\varphi] &=\int \inpd{ a^{-1} \Pi^{-1}  (\dot{\varphi}  - \Pi b(\varphi))} {\d \varphi}
\\
&=\int_0^1  \inpd{ \Pi^{-1}  ({\varphi}' \lambda   - \Pi b(\varphi))} {\varphi'}_a  \d \alpha
\end{split}
\]
where the scalar-valued function  $\lambda:=\d \alpha/\d t\in [0,+\infty]$ is the 
change of variable between the 
physical time $t$ and  the  arc length $\alpha$. Here
 $\dot{\varphi}= d\varphi/\d t$ is the time derivative
 and ${\varphi}' = d\varphi (\alpha)/d\alpha$
 is the tangent vector of the curve for the $\alpha$-parametrization.
To derive the expression of $\lambda$ in terms of $\varphi$ and $\varphi'$, we use 
the  condition  that the Hamilton along the trajectory is the constant zero \cite{Heymann2006,Landau-Lifshitz-Mechanics}.
Plugging in the generalized momentum given in  \leqref{eqn:Gp}, we solve $\lambda$  from the following
zero-valued Hamiltonian,
\[ H(\varphi, p)= \inpd{\Pi b(\varphi)}{ \Pi^{-1}  (\dot{\varphi}  - \Pi b(\varphi))}_a
 + \frac12
\left \| \Pi^{-1}  (\dot{\varphi}  - \Pi b(\varphi))\right \|_a^2 =0.\]
Since $\dot{\varphi} = \varphi' \lambda $, the above equation gives the result that
\[
\begin{split}
\frac12  \lambda^2 \|\Pi^{-1} \varphi'\|_a^2 + \lambda \inpd{\Pi b(\varphi)-\Pi^{-1}\Pi b(\varphi)}{\Pi^{-1}\varphi'}_a \\
+ \frac12 \left \| \Pi^{-1}\Pi b(\varphi) \right \|_a^2 - \inpd{\Pi b(\varphi)}{\Pi^{-1} \Pi b(\varphi)}_a =0.
\end{split}
\]
By the   Proposition \ref{prop:aorth}  and Proposition \ref{cor:atri} ,  it is further simplified as
\[
\frac12  \lambda^2 \|\Pi^{-1} \varphi'\|_a^2 - \frac12 \| \Pi^{-1}\Pi b(\varphi) \|_a^2 =0.
\]
Since $\lambda>0$,  we choose the positive root of the above quadratic equation,
\begin{equation}
\label{eqn:t2s}
\lambda = \frac{\|B\|_a}{\|{\nu}\|_a}
\end{equation}
where   
\[ B:= \Pi^{-1}\Pi b(\varphi), \qquad   \nu:= \Pi^{-1} \varphi' .\]
 
Therefore, we obtain the expression of $\hat{S}$ for $\varphi \in \mathcal{A}$,
\begin{equation}
\label{eqn:gS-form} 
\begin{split}
\hat{S}[\varphi] &=\int_0^1  \inpd{  \lambda   \Pi^{-1} {\varphi}'    -  \Pi^{-1}\Pi b(\varphi)   }    {\varphi'}_a  \d \alpha
  \\
    &=\int_0^1  \inpd{\lambda \nu    - B    }    {\varphi'}_a  \d \alpha
\\
    &=\int_0^1  \inpd{ \frac{\|B\|_a}{\|{\nu}\|_a}
 \nu    - B    }    {\varphi'}_a  \d\alpha\\
 &=\int_0^1    {\|B\|_a}{\|{\Pi^{-1}\varphi'}\|_a}
     - \inpd{ B}    {\varphi'}_a  \d\alpha.
\end{split}
\end{equation}
Here, Proposition \ref{prop:aorth} is used again for the last equality.

%
 \section{Constrained Minimum Action Method}
 
 So far  we have derived two  action functionals on $\mfd$, \leqref{S} on the space of 
 absolution continuous functions $C([0,T],\mfd)$ and  
\leqref{eqn:gS-form} on the space of curves living on $\mfd$.
The variational solutions of these acton functionals will give the minimum action path.
The variational problems are  solved by numerical optimization solver.
We next discuss about the numerical issue of this  variational problem  for $\mfd$ being explicitly specified by constraint functions.

 Let's recall that   the constraints for the system  are specified by non degenerated constraint functions   $c_k(x)=0$, $k=1,2,\ldots, n-d$.
Thus  \begin{equation} \label{eqn:mfdcon} \mfd=\{x\in \Real^n:  c_k(x)=0, k=1,2,\ldots, n-d\}. \end{equation}
 The basis vector for the space $\Ker(\Pi_x)$ is  $\xi_k=\nabla c_k(x)$.
Explicit formula of $\Pi^{-1}$ can be expressed in terms of $\xi_k=\nabla c_k$ ,
following the same procedure as in Proposition  \ref{cor:xi}. Then 
the variational problem for \leqref{S} and  
\leqref{eqn:gS-form}  can be numerically solved under the constraint $\{c_k=0\}$
by any modern optimization  solver.  When the local coordinate representations for  \leqref{S} and  
\leqref{eqn:gS-form} are available for some practical problems, the optimization procedure 
can be performed directly in the local coordinate form.

The calculation of  the minimum action curve for the geometric action functional
involves a reparametrization step which is based on the calculation of  the arc length of the curve $\varphi(\alpha)$.
It may be more natural to use the geodesic distance on $\mfd$ to define the arc length,
but  it is practically convenient to just use  the Euclidean arc length.
If the number of  discrete images in representing the curve is sufficiently large,
these two choices of the distance between neighbouring images measured by geodesic or Euclidean metrics would not give much difference. 

In the following,  we describe  one example of the action functionals  for  the spherical case $\mfd=\mathbb{S}^{d}$ where $d=n-1$.
When $n=3$, this correspond to the first example in next section of a rigid rod model for 
liquid crystal.
The constraint function for $S^{d-1}$ is  $c(x)=\|x\|^2 -1=0$.
The projection onto the tangent space is $\Pi_x = I - \nd(x) \otimes  \nd(x)$ where $\nd(x)=x/\|x\|$ is the unit ($L^2$ norm) normal.
 $\Ker(\Pi_x)=\{\nd(x)\}$ and $\img(\Pi_x)=T_x\mathbb{S}^{n-1}$.
The  calculation in Proposition  \ref{cor:xi} shows that  for any $v\in \img(\Pi)$
\begin{equation*}
 \Pi^{-1} v = v  -  \frac{\inpd{v}{\nd}_a}{\inpd{\nd}{\nd}_a}   \nd
= v  -  {\inpd{v}{\nd_a}_a}   \nd_a
\end{equation*}
and $\|\Pi^{-1} v \|_a^2 = \|v \|_a^2 - { (\inpd{\nd_a}{v}_a)^2} $, 
 where $\nd_a = \nd/\|\nd\|_a$ is the unit vector in sense of $\|\cdot\|_a$ norm.
 Then,  the action functional $S^\mfd_T$ in \leqref{eqn:rate-mfd} 
 becomes
\begin{equation}
\label{eqn:rate-mfd-2}
\begin{split}
S^\mfd_T [\phi]  =
&  \frac12 \int_0^T  
\left \| \dot{\phi}-\Pi b(\phi) \right \|_a^2
-\inpd{\frac{\phi}{\|\phi\|_a}}{\dot{\phi}-\Pi b(\phi)}_a^2\d t.
\end{split}
 \end{equation}
Note that the first term is exactly $S^1_T$ \leqref{eqn:S1}, and $S^\mfd_T[\phi] \leq S^1_T[\phi]$.
Likewise, we have the expression of the geometric action function  
\leqref{eqn:gS-form}  in this case, which is
\[ \begin{split} \hat{S}[\varphi] = \int _0^1 & \sqrt{\left ( \|\Pi b\|^2_a - \inpd{\Pi b}{\varphi}^2_a/\|\varphi\|^2_a\right ) 
\left ( \|\varphi'\|^2_a - \inpd{\varphi'}{\varphi}^2/\|\varphi\|^2_a \right )}  \\
&-  \inpd{\Pi b}{\varphi'}_a  + \inpd{\Pi b}{\varphi}_a \inpd{\varphi'}{\varphi}_a / 
\|\varphi\|^2_a~\mbox{d}\alpha. 
\end{split}\]

 \section{ Examples}
 
 In this section, we apply the constrained minimum action method to study the transition pathways 
 for  the motion of  one class of liquid crystal molecules. 
 This type of macromolecules are usually modelled as rigid rods so the configuration space for each rod  is $\mathbb{S}^2$. More realistic  models such as 
 general bead-rod-spring models have more complex intrinsic constraints for the molecular configurations; the details are well explained  in Chapter 5 of  reference  \cite{OttingerBook}.
 The rigid rod model we are studying here is  the typical  building block for those 
 chain models.
 
 \begin{figure}[htbp]
\begin{center}
\includegraphics[width=0.46\textwidth]{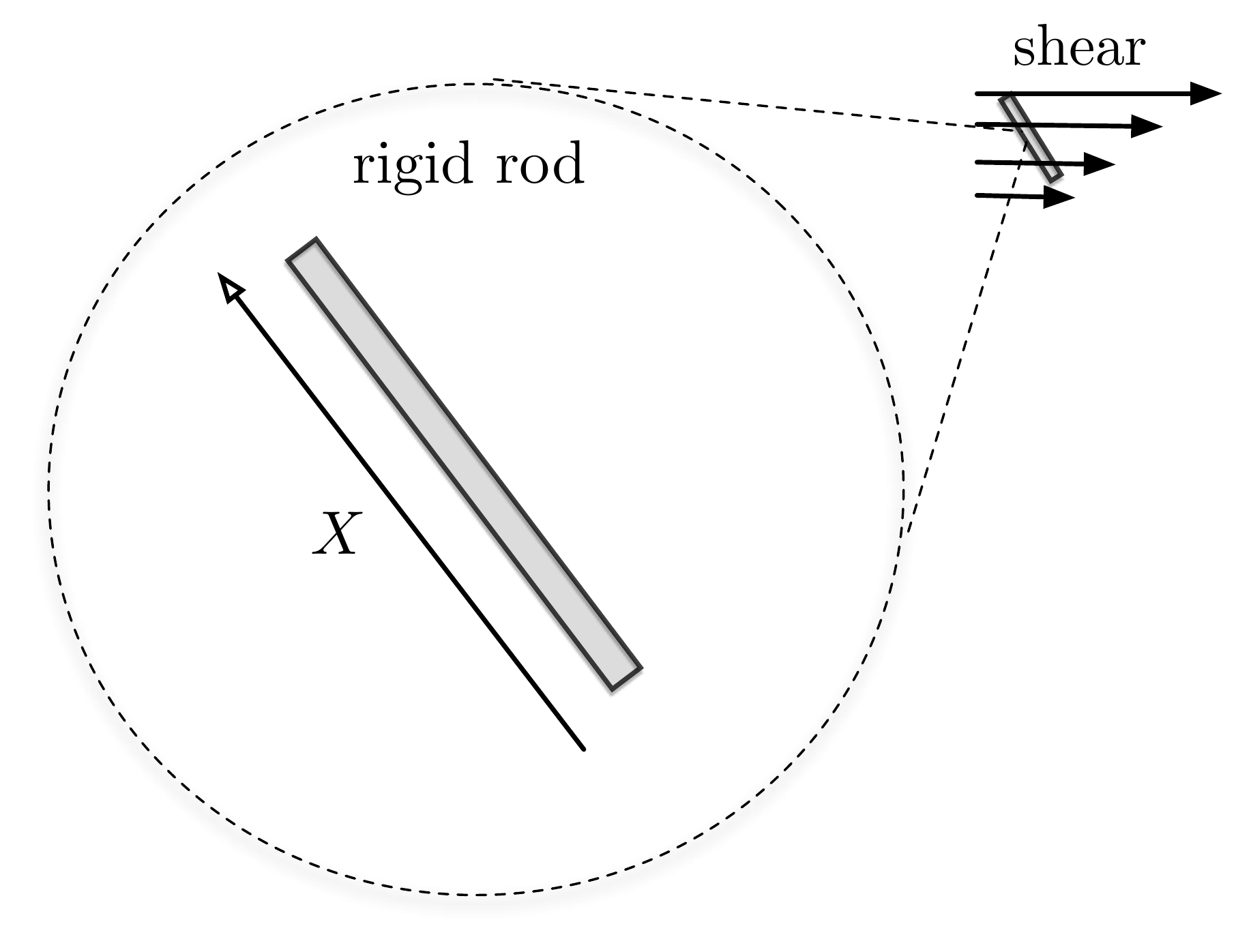}
\caption{Rigid rod polymeric model in shear flow.
The length of the directed vector $\bX$ is one. }
\label{fig:rod-1}
\end{center}
\end{figure}

Typically, there are many equilibrium states for the molecular configurations.
Depending on the interaction between molecules,  there could be  some  spontaneously preferred directions  $\bX$ for the molecules.
In many cases where the ensemble statistics is of interest, the direction $\bX$ and $-\bX$ is undistinguished due to symmetry.
But at microscopical level,  each   individual configuration does switch between the symmetric two metastable states $\bX$ and $-\bX$.
When these macromolecular polymers are added into solvent (Figure \ref{fig:rod-1}), then
the mixed solution  has  interesting hydrodynamical and rheological features different from the Newtonian fluid.
The study of complex fluid mainly focuses on  macroscopic quantities of   polymeric fluid, such as   viscoelasticity.
However, the change of   macromolecular  configurations  at the microscopic level due to   thermal fluctuation and   fluid shear is of its own interests,
in particular when these macromolecules, for instance liquid crystals,  are directly responsible for some   physical mechanisms in practice
 such as colour control for display devices. 

In the following, we present two examples to understand the transition paths in the rigid rod model.
 In the  first example, we study the single rod molecule with   quadratic potential in shear flow.
  Due to spheric symmetry, any  linearly stable state $\bX$ has a symmetric stable one $-\bX$.
  The transition from $\bX$ to $-\bX$ corresponds to the flip over  process of the rod molecule.
 How the shear rate impacts this flip over process  is of our interest. 
 Our second example includes two rods  with interaction between them.
This is the simplest case for the weakly interaction particle system \cite{xiang2004}.  
To see how the anisotropic  diffusion tensor play roles in transition path,
 we artificially  assign two different  diffusion coefficients,  $\sigma_1$ and $\sigma_2$,  for the two rods and 
 investigate the effect of the ratio  $\sigma_2/\sigma_1$ on the transition pathways.
 Although the diffusion coefficient  (i.e., temperature) of  two molecules seem to be the same in physical reality,
 our manipulation of anisotropic noise in this model  produces some interesting results, which could 
 be instructive in the  general case of the state-dependent noise $\sigma(x)$ and may be quite useful 
 when the precise control of noise size  for each individual rod (or two groups of rods) is possible.  
  
  Lastly, we remark that we only report the results  from the constrained minimum action method 
  based on  the geometric  action formulation. Thus, the objects we calculated are curves in the 
  phase space.  The pathways from the constrained \FW action functional are consistent with these results 
  when the underlying time interval is sufficient large.
  
 \subsection{Flip over process  of one rigid rod  }

Consider a unit sphere $\mathbb{S}^2$ in $\Real^3$.     $X=(X_1,X_2,X_3)\in \Real^3$.  
 Let $V: \Real^3 \to \Real$ be the potential energy with symmetry $V(x)=V(-x)$ and 
  $W_t$ be a Brownian motion in $\Real^3$.
Write the normal vector  $\nd(x)=x / \|x\| \in \mathbb{S}^2$. The motion of the rod molecule 
in consideration is described by the following equation,
\[  \d X = (I-\nd(X) \nd(X)^\transpose ) \bigg( (-\nabla V(X) +  K_0 X) \d t +  \sqrt{\eps} \circ \d W_t \bigg)
 \]where   $K_0$ is the matrix of the shear rate tensor in the Cartesian coordinate.

Here the noise is isotropic and the manifold $\mfd $ is $\mathbb{S}^2$. The  
 geometric action functional in   \leqref{eqn:gS-form}  is  reduced to 
 \begin{equation}
\begin{split}
\hat{S}[\varphi]  &=\int_0^1    {\|\Pi^{-1}\Pi b(\varphi)\|}{\|{\Pi^{-1}\varphi'}\|}
     - \inpd{ \Pi^{-1}\Pi b(\varphi)}    {\varphi'} d\alpha
     \\ 
     &=\int_0^1    {\|\Pi b(\varphi)\|}{\|{\varphi'}\|}
     - \inpd{ \Pi b(\varphi)}    {\varphi'} d\alpha
\end{split}
\end{equation}
where  $ {\|\varphi(\alpha)\|=1} $ for all $\alpha$.

 We assume the following quadratic form of the external potential function $V$
 \begin{equation}
 \label{eqn:qV}
 V(x) = \sum_{i=1}^3 \frac12 \mu_i x_i^2, \mbox{ where } \mu_3>\mu_2>\mu_1>0.
 \end{equation}
 The two local minima  of $V$ on $\mathbb{S}^2$ are 
$\vec{e}^{(1)}=(1,0,0)$ and  $-\vec{e}^{(1)}=(-1,0,0)$; the two local maxima 
are $\vec{e}^{(3)}=(0,0,1)$ and $-\vec{e}^{(3)}=(0,0,-1)$; the saddles are $\vec{e}^{(2)}=(0,1,0)$
and $-\vec{e}^{(2)}=(0,-1,0)$.   In  the example below, we simply set  $(\mu_1,\mu_2,\mu_3)=(1,2,3)$.

For the quadratic potential \leqref{eqn:qV},  the SDE  then becomes  the following form 
\begin{equation}
\label{eqn:linSDEonS2}
  \d X = (I-n(X) n(X)^\transpose ) ( KX \d t +  \sqrt{\eps} \circ \d W_t )
  \end{equation}
where $K=\mbox{diag}\{\mu_1,\mu_2,\mu_3\}  + K_0$.
We consider two forms of shear rate matrix $K_0$
 corresponding to different directions of the shear flow.

\subsubsection{Shear flow: example 1}
\label{sssec:shex1}
We first consider the following shear flow where $x_1$ is the streamwise direction,
$x_2$ is the shearwise direction
and $x_3$ is the spanwise direction. So it is  assumed that
  \begin{equation}
  \label{eqn:sh12}
 K_0 =  \begin{bmatrix}0 &\dot{\gamma}_{12}&0 \\ 0  & 0 & 0 \\  0  & 0 & 0 \end{bmatrix}.
 \end{equation}
 Here  the shear rate $\dot{\gamma}_{12}$ is a constant parameter.

The deterministic drift  flow on $\mathbb{S}^2$ is $\dot{X}=(I-\nd \nd^\transpose )KX$.
 The fixed points of this  flow are the following three   vectors on $\mathbb{S}^2$
\begin{eqnarray}
\nd^{(1)}&=&(1,0,0)^\transpose, \notag \\
\nd^{(2)}&=&(-\dot{\gamma}_{12}, \mu_2-\mu_1,0)^\transpose/\sqrt{\dot{\gamma}_{12}^2+(\mu_2-\mu_1)^2},  \label{eqn:sa}\\
\nd^{(3)}&=&(0,0,1)^\transpose, \notag
\end{eqnarray}
and their symmetric counterparts $-\nd^{(i)},\, i=1,2,3$.
In total, there are three pairs of fixed points. Since $\mu_3>\mu_2>\mu_1>0$ in the quadratic potential (\leqref{eqn:qV}), we can
derive the following linear stability results for infinitesimal perturbations.  
The pair $\pm \nd^{(1)}$ is   linearly stable ( classified as  { \it sink } and denoted as $si_+$ and $si_-$, respectively) with two unstable eigen directions
$\vec{e}^{(2)} $ and $\vec{e}^{(3)}  $.
The pair of   $\pm \nd^{(3)}$ is   linearly unstable (classified as {\it source} and denoted as $so_+$ and $so_-$, respectively).
The pair of   $\pm \nd^{(2)}$ is saddle point (and denoted as $sa_+$ and $sa_-$, respectively) with one stable eigen direction $\vec{e}^{(3)}$
(the unstable eigen direction relies on $\dot{\gamma}_{12}$).
The separatrix on the unit sphere between the two sources  $si_+$ and $si_-$ is the great circle  of $\mathbb{S}^2 $in  the plane spanned by 
$sa_\pm$ and $so_\pm$.

The introduction of the shear rate in form of \leqref{eqn:sh12} only affects the orientation of the saddle point (\leqref{eqn:sa}).
The positive value of  shear rate $\dot{\gamma}_{12}$ has the effect of rotating the 
saddle direction $\nd^{(2)}$ counterclockwise (looking down from $x_3$-direction, i.e.,  vertical  direction).
The negative $\dot{\gamma}_{12}$  gives the opposite rotation direction.

We are  concerned with the flip over process of the rigid rod, 
i.e.,  the transition between two symmetric stable fixed points  $si_+=\nd^{(1)}$ and $si_-=-\nd^{(1)}$. The minimal action for this transition is related to the 
 frequency of this process (  $\propto \exp(-\inf S/\eps)$ \cite{FW1998}). 
 The smaller the  minimal action, the more frequently the rod flips between   two stable states.

To resolve all possible minimizers of the variational problem $\inf_{\varphi\in \mathbb{S}^2} \hat{S}[\varphi]$,
the initial guesses of the path should be carefully constructed. The  idea of setting initial guess is as follows.
Since on the separatrix  between $si_+$ and $si_-$, there are four  fixed points,  $so_\pm$
and $sa_\pm$.  We then construct the different initial paths passing through these points, respectively.
In consideration of the symmetry for the case of $so_\pm$,  we only need to test three different initial guesses,
which give three different local minima  of the action functional $\hat{S}$.
As a result, the obtained three minima correspond to the minimal actions from $si_+$ to saddles $sa_-$, $sa_+$, and  
  $so_-$ (or $so_+$), respectively. 
  The  minimum among these three  minimized actions gives the global optimum and thus corresponds to the correct transition path between $si_-$ and $si_+$.
  Refer to Figure \ref{fig:sd_sh_vs_act-1} for the plot of these three actions when the shear rate is varied.
 This evidence shows that   the shear of the flow field lowers the global minimum of the action,  hence increase the 
 flip over frequency.  At  a high shear rate,  the frequency could be so large that the rod molecule
 would oscillate  between the direction $si_-=-\nd^{(1)}$ and $s_+=\nd^{(1)}$. 
 
\begin{figure}[htbp]
\begin{center}
\includegraphics[scale=0.38]{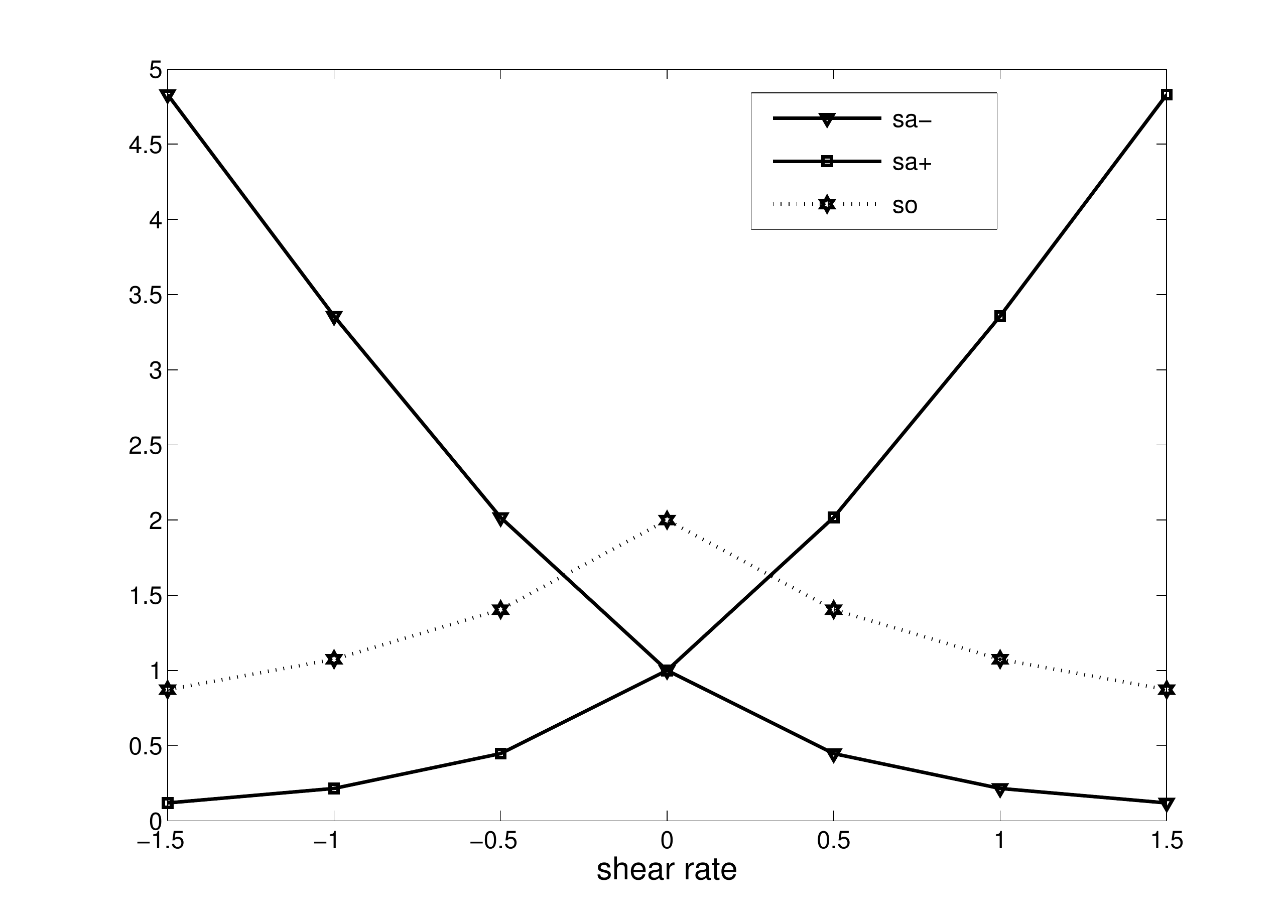}
\caption{ The minimum actions (vertical axis) for three paths from $si_+$ to $sa_-$, $sa_+$ and $so_\pm$, respectively.}
\label{fig:sd_sh_vs_act-1}
\end{center}
\end{figure}
%
%
%

\begin{figure}[htbp]
        \centering
        \begin{subfigure}[b]{0.5\textwidth}
                \includegraphics[width=\textwidth]{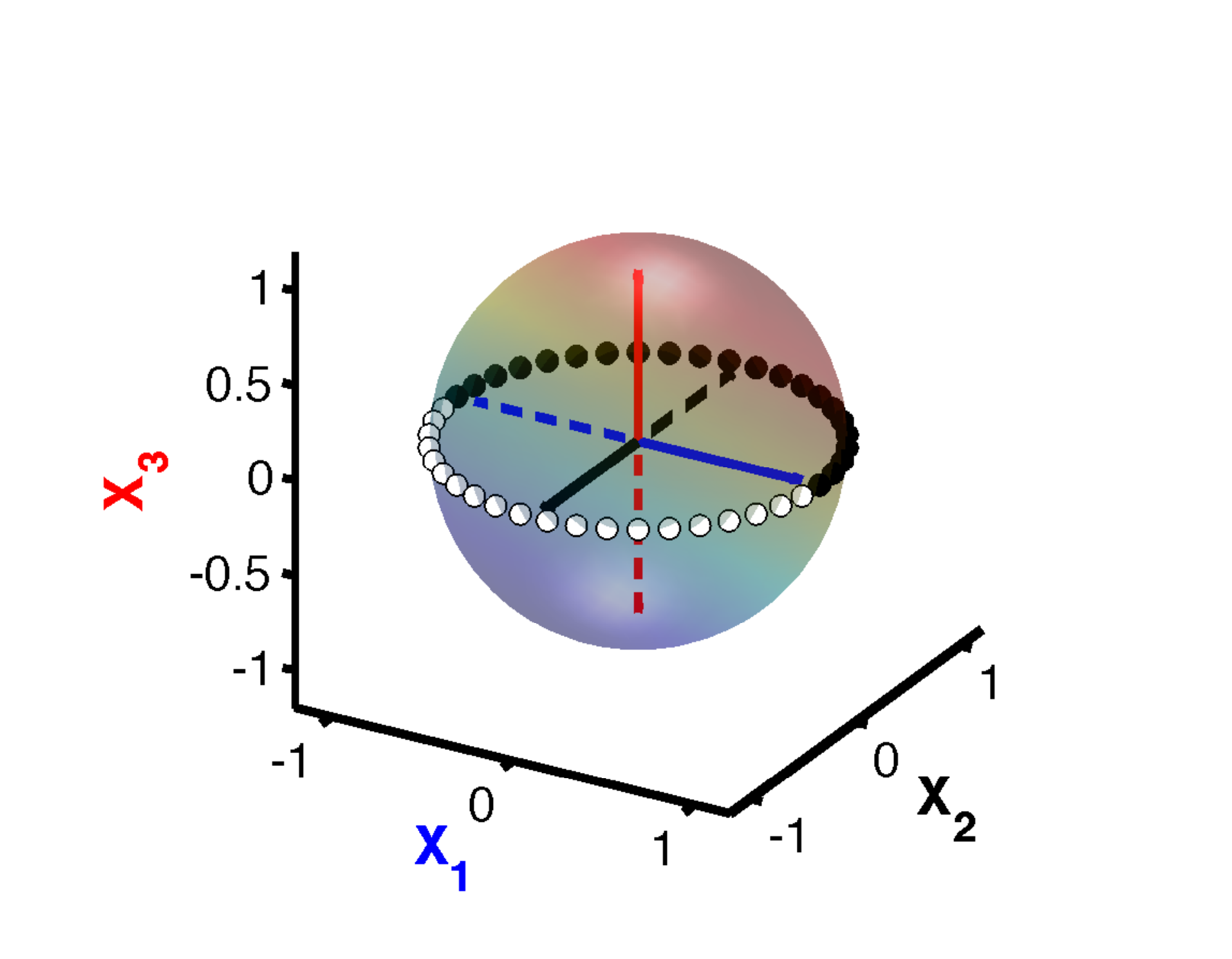}
              \caption{shear rate $\dot\gamma_{12}=0$. }
\label{fig:sd_sh_00}
        \end{subfigure}%
        ~ 
          
        \begin{subfigure}[b]{0.5\textwidth}
                \includegraphics[width=\textwidth]{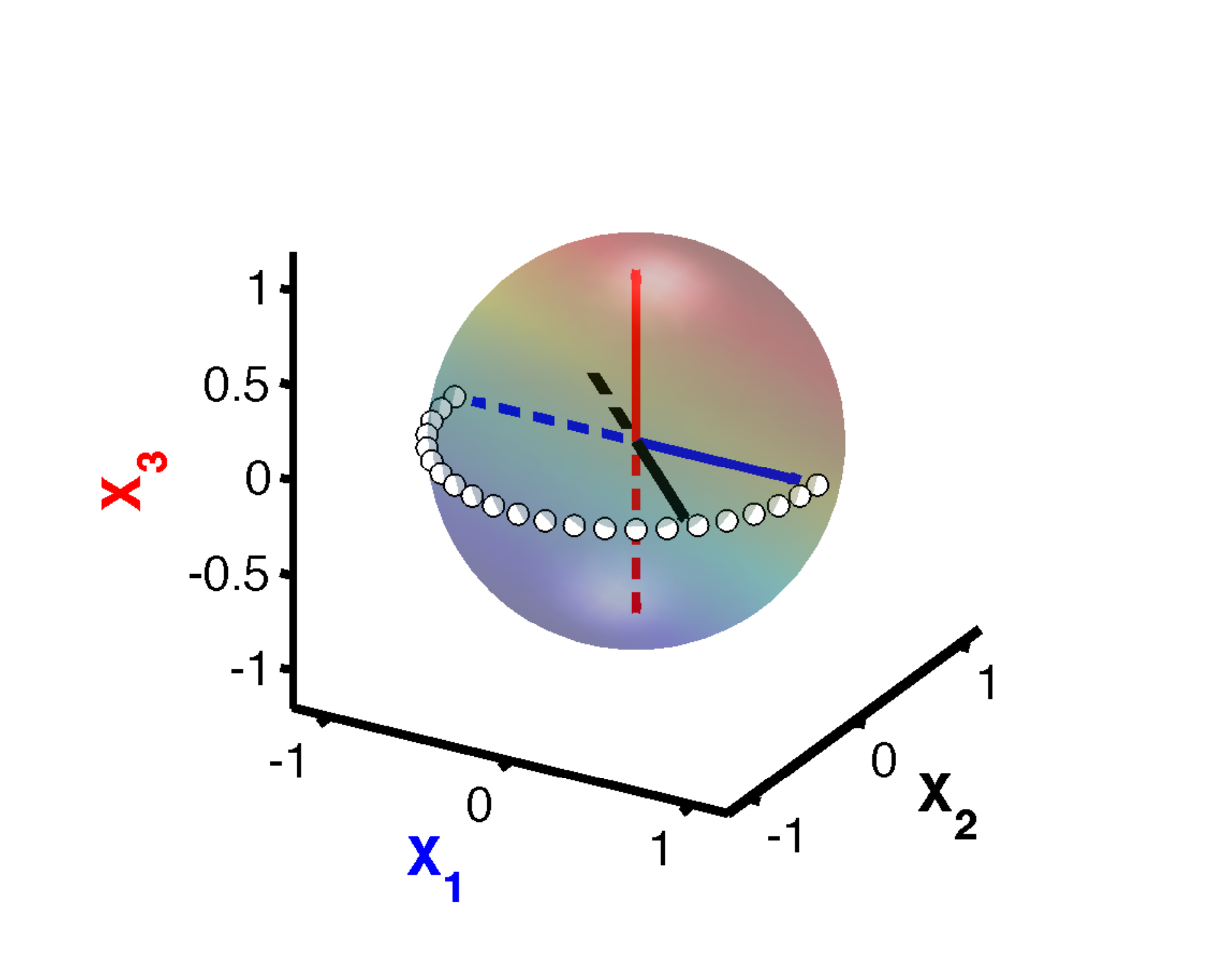}
\caption{shear rate $\dot\gamma_{12}= 1.0$}
\label{fig:sd_sh_1_g12}
        \end{subfigure}
        ~ 
        \begin{subfigure}[b]{0.5\textwidth}
                \includegraphics[width=\textwidth]{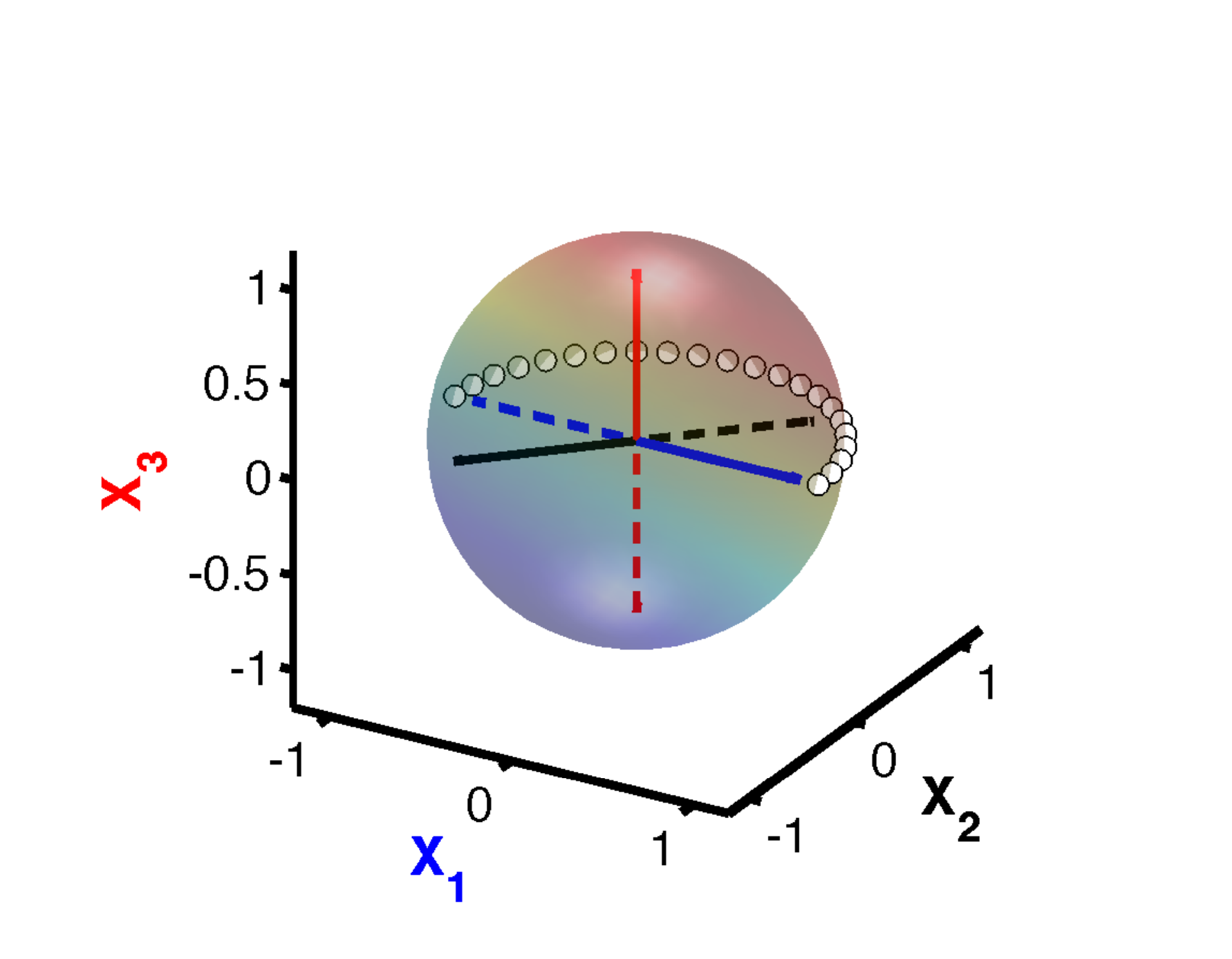}
\caption{shear rate $\dot\gamma_{12}= -1.0$}
\label{fig:sd_sh_n1_g12}
        \end{subfigure}
        \caption{Global minimum action paths for different shear rates. The two symmetric filp-over paths passing $sa_-$ are  shown in  white and black, respectively.}\label{fig:sh-path-1}
\end{figure}

Among the obtained three paths from different initial guesses, the global minimum is 
the one passing the saddle $sa_-$ or $sa_+$, depending on the direction of the shear,
i.e., the sign of $\dot{\gamma}_{12}$. 
These paths are one of semi great circles entirely in the $x_1$-$ x_2$ plane.
Figure \ref{fig:sh-path-1} shows the global minimum action path starting from $si_+$ for $\dot{\gamma}_{12}=0,1,-1$.
For instance, when $\dot{\gamma}_{12}>0$, the saddle $sa_-$ (the solid black line) is shifted closer to $si_+$ so that 
it takes less action for the system to escape from $si_-$ to the separatrix by selecting this saddle $sa_-$.
A similar picture holds for negative $\dot{\gamma}_{12}$ where 
the saddle $sa_+$ (the dashed black  line) is shifted closer to $si_+$.

\subsubsection{Shear flow: example 2}

Next we study the transitions with the following shear rate tensor
 \[
 K_0 = \left[ \begin{array}{ccc} 0 &0 &\dot{\gamma}_{13} \\   0 & 0&0  \\   0 & 0 & 0 \end{array}\right].
 \]
The fixed points for this $K_0$ become 
\begin{eqnarray*}
\nd^{(1)}&=&(1,0,0)^\transpose, \\
\nd^{(2)}&=&(0,1,0)^\transpose, \\
\nd^{(3)}&=&(-\dot\gamma_{13}, 0, \mu_3-\mu_1)^\transpose/\sqrt{\dot\gamma_{13}^2+(\mu_3-\mu_1)^2}.
\end{eqnarray*}
and $ si_-=- \nd^{(1)}$ and $si_+=\nd^{(1)}$ are sinks,
$sa_- = - \nd^{(2)} $ and $sa_+=\nd^{(2)}$ are saddles,
$so_-=-\nd^{(3)}$ and $so_+=\nd^{(3)}$ are  sources.
The heteroclinic orbits among  these fixed points are   similar to the  previous example in \S\ref{sssec:shex1}: 
They are the great circles connecting the neighbouring fixed points.
The separatrix between $si_+$ and $si_-$ is also the great circle in the plane of $sa$ and $so$.
The difference from the example  in \S\ref{sssec:shex1} is that now the shear rate affects the location of the sources $so_\pm$
while the saddles $sa_\pm$ are unchanged. 

Again, we are interested in the transition from $si_+$ to $si_-$ and shall  examine the minimum action paths
with different initial guesses which pass through the fixed points
$so_-$, $so_+$ and $sa_\pm$, respectively.  Figure \ref{fig:sd_sh_vs_act-2} shows the minimum actions 
for these three paths.  From this figure, we can observe that a larger   shear rate deceases the actions
both to the saddle and to the source.  However, there is a competition between these two local minima of the action.  When  the shear rate is small,  the path passing the saddle is the global solution. 
But when the shear rate is very large,  the calculation shows that the action to the source can be slightly smaller than the one
to the saddle so that the transition state changes from the saddle to the source.
This suggests that there is a bifurcation point of the parameter $\dot{\gamma}_{13}^*$ (around $1.9$ for this example in our calculation) 
for the patterns of the global minimum action path.

\begin{figure}[htbp]
\begin{center}
\includegraphics[width=0.45\textwidth]{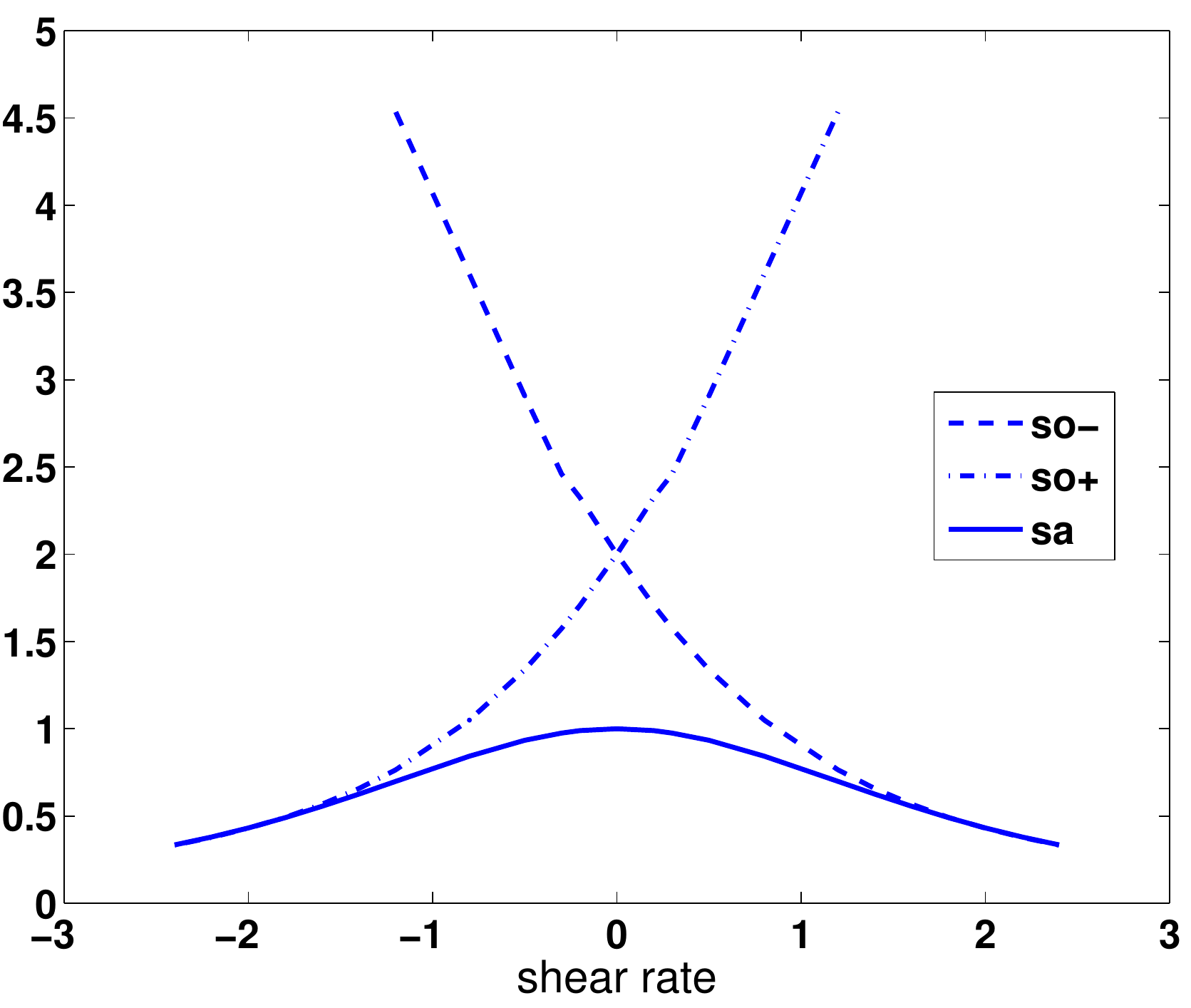}\quad
\includegraphics[width=0.48\textwidth]{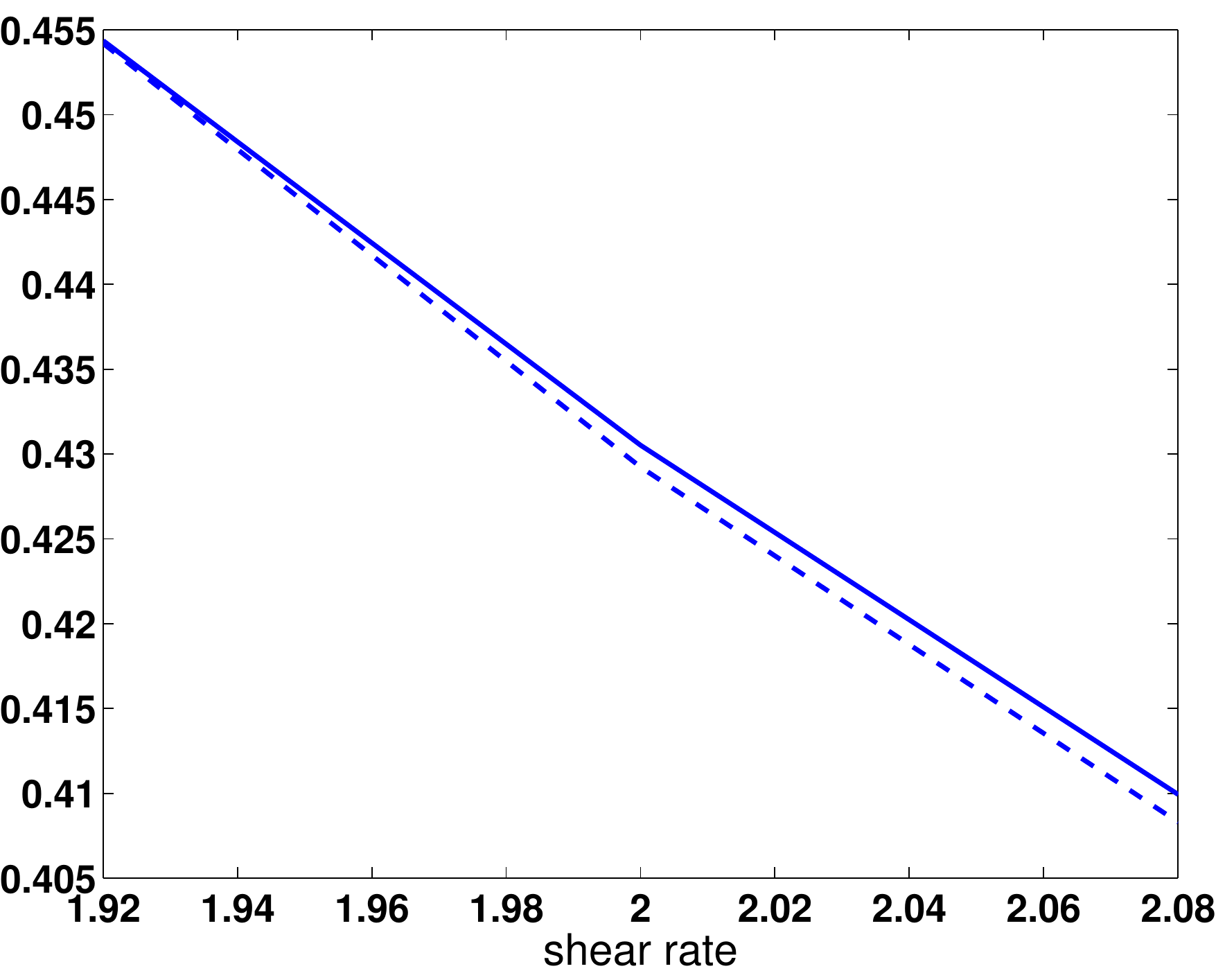}
\caption{The minimum actions for three paths from   $si_+$  to $so_-$, $so_+$ and $sa_\pm$ (either $sa_-$ or $sa_+$ since the minimal actions are the same due to symmetry), respectively.
When $\dot{\gamma}_{13}$ passes the critical value $\dot{\gamma}^*_{13}\approx 1.9$,
the  global minimum path changes from passing  o $sa$ to  passing $so_-$, i.e., the transition state changes from $sa$ to $so_-$.
The right panel is the zoom of the left panel for a window near $\dot{\gamma}_{13}=2$. }
\label{fig:sd_sh_vs_act-2}
\end{center}
\end{figure}

\begin{figure}[htbp]
        \centering
        \begin{subfigure}[b]{0.5\textwidth}
                \includegraphics[width=\textwidth]{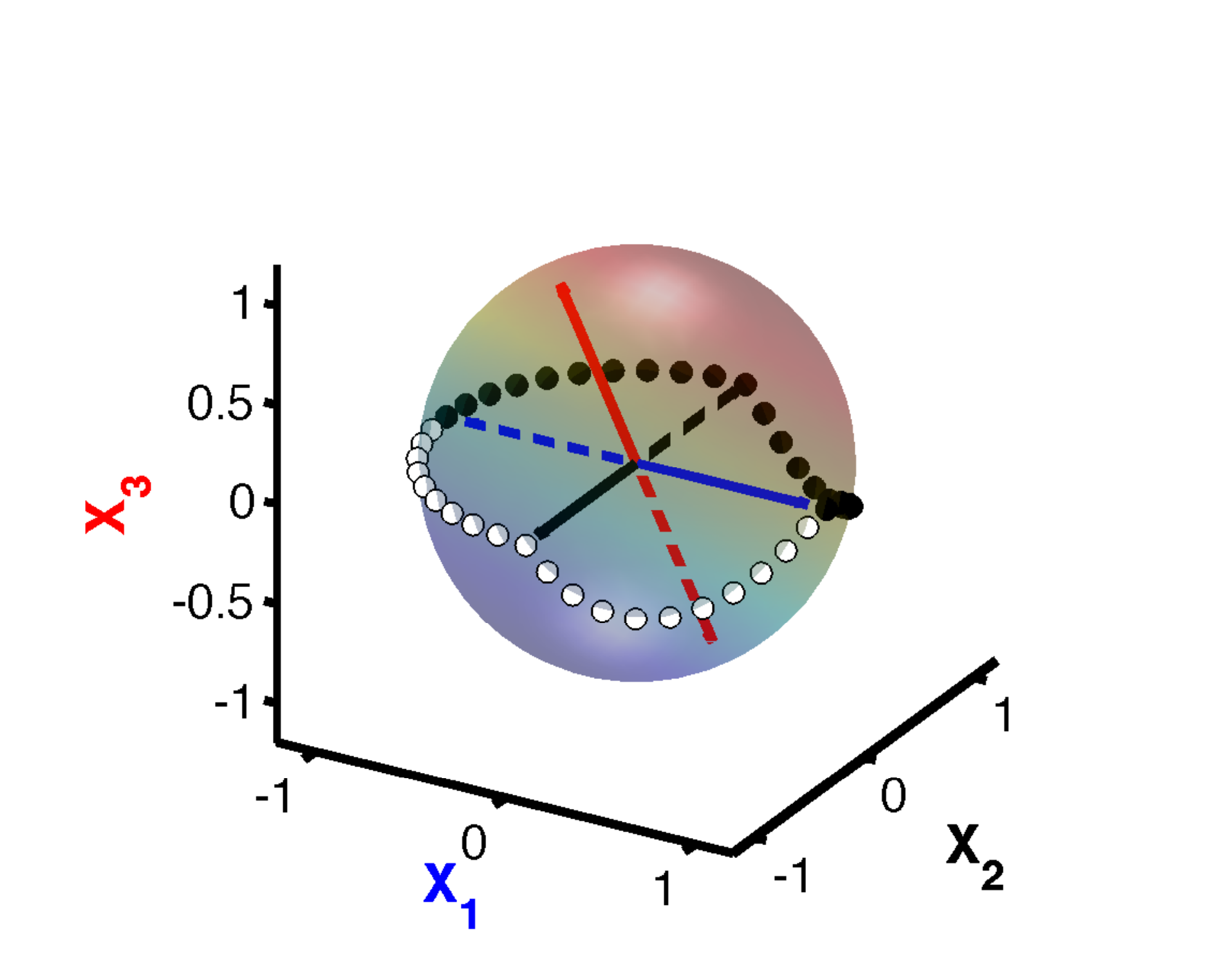}
              \caption{shear rate $\dot{\gamma}_{13} =1.0.$}
\label{fig:sd_sh_1_g13}        \end{subfigure}%
        ~ 
        \begin{subfigure}[b]{0.5\textwidth}
                \includegraphics[width=\textwidth]{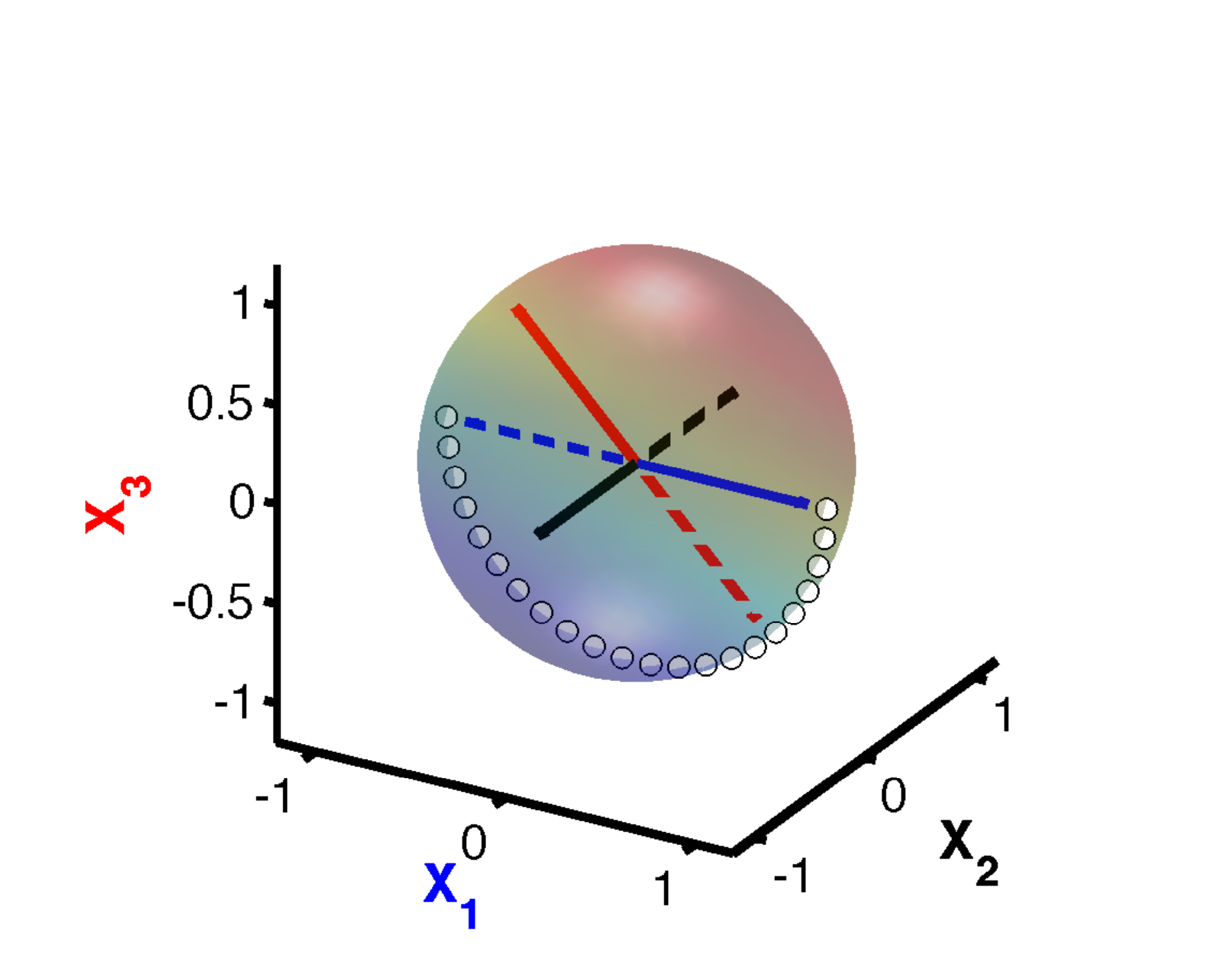}
\caption{shear rate $\dot{\gamma}_{13}= 2.0$}
\label{fig:sd_sh_n1_g13}
        \end{subfigure}
        \caption{Global minimum action paths for different shear rates.
        (a):    The transition path from $si_+$ to $si_-$ through $sa_-$ and its symmetric mirror 
are both shown.
(b) The path is the  semi-circle in the plan spanned by $so$ and $si$ ($x_1$-$x_3$ plane) .
The initial guess of the path in the minimum action method is the path in Figure \ref{fig:sd_sh_1_g13}. }
        \label{fig:sh-path-13}
\end{figure}

The above conclusion can be better understood if we plot the global minimum action path
for $\dot{\gamma}_{13}=1$ and $2$ in Figure \ref{fig:sh-path-13}.
The positive value of  shear rate $\dot{\gamma}_{13}$ has the effect of tilting the 
unstable fixed point $so_+=\nd^{(3)}$ (the solid red line) counter-clockwisely in the $x_1$-$x_3$ plane (looking  from $-x_2$-direction).
Such tilts will pull  $so_-$ (the dashed red line) towards  $si_+$ (the solid blue line) and push $so_+$ away from $si_+$.
  However, when the shear rate is not strong, this push is not significant enough to beat the action
of the path through the saddle $sa_\pm$ (the pair of curves shown  in Figure \ref{fig:sd_sh_1_g13} ).
 When $\dot{\gamma}_{13}$ continues to increase by passing the critical value $\dot{\gamma}^*_{13}$, 
  the shear-induced tilt will become strong enough to 
   lower the action  to  reach  $so_-$  significantly so as to  become a  global soluiton.

In summary, when the shear of the fluid affects the unstable fixed points $so_\pm$ of the molecular configuration on $\mathbb{S}^2$, 
the competition of the minimum action paths passing through the saddle $sa_\pm$ 
or through  the source $so_\pm$ would generate a bifurcation of the patterns of the global path.
The same phenomena have been observed before,  for  instance, in some  planer  (non-gradient) system  \cite{Maier1993PRE}. 
For real problems, the shear rate tensor $K_0$ may be the combination of the above two examples we have studied; from the analysis 
above, we  expect that the similar  bifurcation of the pathways could happen for different size of the shear rate.
It is also generally believed that the shear would  lower the global minimum  action, thus increase the flip over frequency.

\subsection{Flip over  of two rigid rods }

 Here, we study a slight generalization of the previous studied single rod case,
a toy model of two interactive  rigid rods. 
Let $\vec{X}_1, \vec{X}_2$ be  the directed unit vector of two rods. We consider the following 
 stochastic dynamics  on $\mathbb{S}^2\times\mathbb{S}^2$,  
\begin{equation}
\left\{
\begin{aligned}
\d \vec{X}_1 &= (I-\vec{n}_1\vec{n}_1^\T)(-\nabla V(\vec{X}_1)\d t -\nabla_{\vec{x}_1}U(\vec{X}_1,\vec{X}_2)\d t + \sigma_1\sqrt{\epsilon}\circ \d\vec{W}_1),\\
\d \vec{X}_2 &= (I-\vec{n}_2\vec{n}_2^\T)(-\nabla V(\vec{X}_2)\d t -\nabla_{\vec{x}_2}U(\vec{X}_1,\vec{X}_2)\d t + \sigma_2\sqrt{\epsilon}\circ \d\vec{W}_2),
\end{aligned}
\right.
  \label{eq:dyn_fun}
\end{equation}
where $\nd_i=\bX_i/\|\bX_i\|, \, i=1,2$. $\sigma_1$ and $\sigma_2$ are two positive  constants.
Here   $U(\vec{x}_1,\vec{x}_2): \mathbb{S}^2\times\mathbb{S}^2\to\Real$ describes the 
interactions of these two rods. 
One common choice of this potential $U$ is the following Maier-Saupe  potential 
\begin{equation}
\label{eqn:U}
 U(\vec{x}_1,\vec{x}_2)=A\sin^2(\theta-\theta_0) 
\end{equation}
where $\theta$ is the angle between $\vec{x}_1$ and $\vec{x}_2$ (Figure \ref{fig:rod-2}),  $A$ is a positive number and  $\theta_0$ is the preferred angle.
We assume $\theta_0=0$ without loss of generality.

\begin{figure}[htbp]
\begin{center}
\includegraphics[scale=0.4]{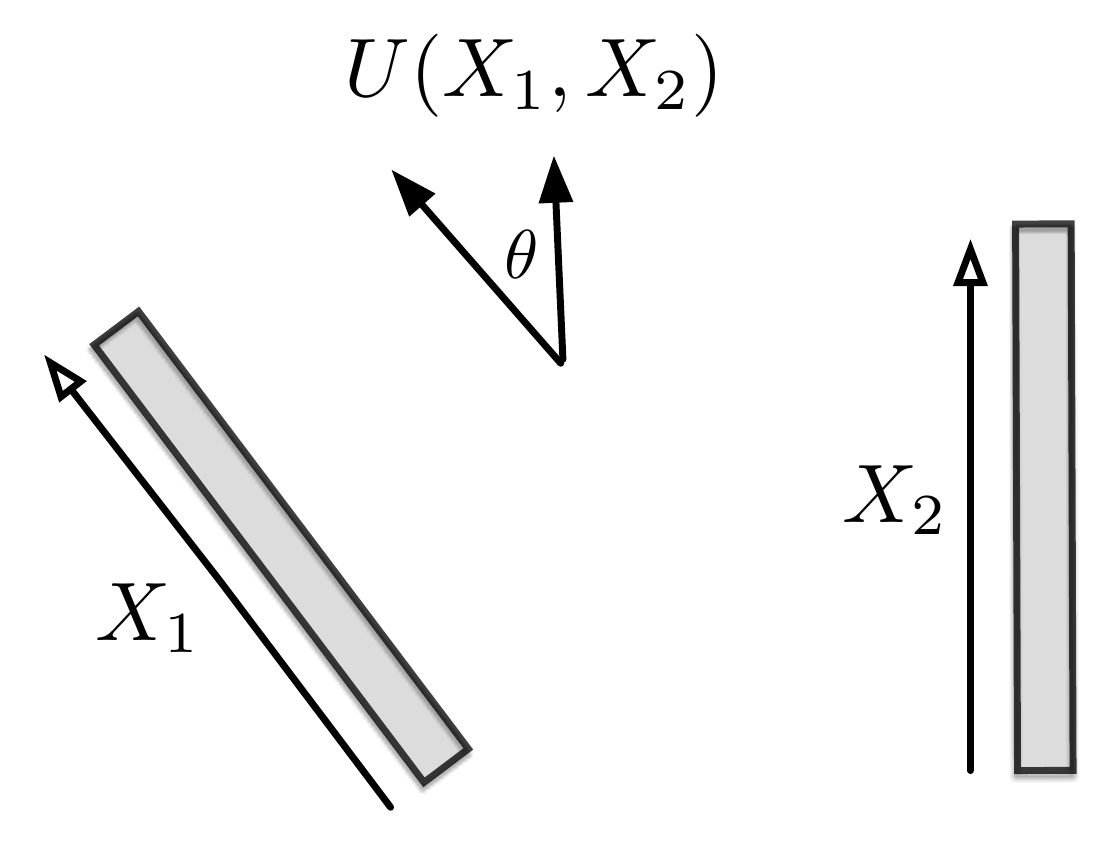}
\caption{Two rigid rods with interaction potential $U$. }
\label{fig:rod-2}
\end{center}
\end{figure}

The model we are studying in  \leqref{eq:dyn_fun}  has no effect of shear flow and is a reversible system when $\sigma_1=\sigma_2$.
In the following, we are interested in how different values of the ratio  $\sigma_2/\sigma_1$ affects the transition paths.
First we give the action functional form for this example. 
We write the path as a pair  $\varphi=[\varphi_1,\varphi_2]\in\Real^3\times\Real^3$.
 Denote   $b(\varphi)=[b_1(\varphi),b_2(\varphi)]\in\Real^3\times\Real^3$
where $b_i(\varphi) = -\nabla V(\varphi_i) -\nabla_{\bx_i} U(\varphi_1,\varphi_2) $ corresponds to the  rod $i$.
The geometric  action functional \eqref{eqn:gS-form} for \leqref{eq:dyn_fun} thus has the following form  
\begin{equation}
\label{eq:AF}
  \hat{S}[\varphi] = \int_0^1\sqrt{\frac{\norm{b_1}^2}{\sigma_1^2} + \frac{\norm{b_2}^2}{\sigma_2^2}}\sqrt{\frac{\varphi_1'^2}{\sigma_1^2}+\frac{\varphi_2'^2}{\sigma_2^2}} - \frac{\inpd{b_1}{\varphi_1'}}{\sigma_1^2} - \frac{\inpd{b_2}{\varphi_2'}}{\sigma_2^2}\d \alpha.
\end{equation}
The constraint is $\norm{\varphi_1}=\norm{\varphi_2}=1$.

\bigskip

We   choose the quadratic potential as in the previous example of one rod.
 $V(\vec{x}) = \vec{x}^\T K\vec{x}$/2.  Here $K=\mbox{diag}\{\mu_1,\mu_2,\mu_3\}$ where $\mu_1<\mu_2<\mu_3$. 
 Next, we  show the following property of the drift flow of \leqref{eq:dyn_fun} on $\mathbb{S}^2\times\mathbb{S}^2$
 for the weak strength of the interaction.

\begin{proposition}
If the coupling constant $A$ in the potential \leqref{eqn:U}  satisfies  \begin{equation}
    \label{eq:cond}
 A <  \frac14    \min\limits_{i\neq j}\abs{\mu_i-\mu_j}
  \end{equation}
hold, then all fixed points of the deterministic drift flow of \leqref{eq:dyn_fun}  are the following $36$ points 
\[ (\pm \vec{e}_i,\pm \vec{e}_j),  ~~~ i,j=1,2,3,\]
where $\vec{e}_i$ is the unit eigenvector of $K$ for eigenvalue $\mu_i$, for instance $\vec{e}_1=(1,0,0)$.
 Moreover, the four points $(\pm \vec{e}_1,\pm \vec{e}_1)$ are stable (classified as {\it sink}),
the four points $(\pm \vec{e}_3,\pm \vec{e}_3)$ are unstable (classified as {\it source}) and other fixed points are all saddles. 
\end{proposition}
\begin{proof}  It can be verified that any fixed point $(\bx_1,\vec{x}_2)$ must satisfy the following equations 
  \begin{gather}
    -K\bx_1+2A\inpd{\bx_1}{\bx_2}\bx_2-2A\inpd{\bx_1}{\bx_2}^2\bx_1 + (\bx_1^\T K\bx_1)\bx_1 = 0,\label{eq:1}\\
    -K\bx_2+2A\inpd{\bx_1}{\bx_2}\bx_1-2A\inpd{\bx_1}{\bx_2}^2\bx_2 + (\bx_2^\T K\bx_2)\bx_2 = 0,\label{eq:2}\\
    \norm{\bx_1}=\norm{\bx_2}=1.
  \end{gather}
  
If $\inpd{\bx_1}{\bx_2}=0$,  Eqs. \eqref{eq:1} and \eqref{eq:2}  suggest  $\bx_1$ and $\bx_2$ must be unit eigenvectors of $K$
corresponding to distinctive eigenvalues, respectively. It gives 24 fixed points $(\pm \vec{e}_i, \pm \vec{e}_j)$ for $i\neq j$ in this case.

If $\inpd{\bx_1}{\bx_2}\neq 0$, Eqs. \eqref{eq:1} and \eqref{eq:2} together imply that 
$\inpd{\bx_1}{\bx_2}$ $(\bx_1^\T K\bx_1-\bx_2^\T K\bx_2)=0,$
or, $x_1^\T Kx_1=x_2^\T Kx_2=\lambda$. Furthermore, by considering \eqref{eq:1} $\pm$ \eqref{eq:2},
we have $\bx_1\pm \bx_2$ are either zero vector or   an eigenvector of $K$. 
The former case gives the other 12 fixed points $(\pm \vec{e}_i, \pm \vec{e}_i)$.
The latter case that $\bx_1 \pm \bx_2$ is an eigenvector of $K$ will eventually
lead to an equality  $\mu_i-\mu_j=4A\inpd{x_1}{x_2}$. But since it follows
$\abs{\mu_i-\mu_j}=\abs{4A\inpd{x_1}{x_2}}\le 4A\norm{x_1}\norm{x_2}=4A,$
which contradicts to condition \eqref{eq:cond}, there are no  other solutions.

The conclusions of the linear stability are based on the calculation of the Jacobian matrices  at these fixed points. We neglect the details.
\end{proof}


In all, there are 36 different fixed points. 
From the above proof we know that $(\vec{e}_i, \vec{e}_j)$ is a fixed point even without the condition \eqref{eq:cond}. 
If the condition \eqref{eq:cond} does not hold, there may be other fixed points and it  can be shown that there are at most 60 fixed points. 
In our numerical calculations,  we choose $K={\rm diag}\{1,3,5\}, A=0.4$ to satisfy the  condition \eqref{eq:cond}.
In addition, we always let $\sigma_1=1$ but allow $\sigma_2$ to vary.

The transition path we will study is from the initial state $(\vec{e}_1,-\vec{e}_1)$ to the final state $(-\vec{e}_1,\vec{e}_1)$, in which  
both rods flip over their initial directions.  Since the initial and final states both lie in the $\vec{e}_1$-$\vec{e}_2$ plane for each rod, then by symmetry consideration, the transition paths, i.e., the minimizers  of the action functional  \leqref{eq:AF}
must also lie in this plane.  Our numerical calculations based on $\mathbb{S}^2\times \mathbb{S}^2$ indeed verify this fact. Therefore, we can visualize the obtained paths and interpret our results on a  lower dimensional product space  $\mathbb{S}^1\times\mathbb{S}^1$. 
It is convenient to use local coordinates $(\theta_1,\theta_2)\in [0, 2\pi)\times [0,2\pi)$ to denote a point  of the path $(\phi_1,\phi_2)$:
\[ \phi_1=[\cos\theta_1,\sin\theta_1,0],\quad \phi_2=[\cos\theta_2,\sin\theta_2,0]. \]
In this local coordinates representation, the initial and final states $(\vec{e}_1, -\vec{e}_1)$ and $(-\vec{e}_1, \vec{e}_1)$ can be written as $(\theta_1,\theta_2)=(0,\pi)$ and $(\pi,0)$, respectively. There are 16 fixed points on $\mathbb{S}^1\times\mathbb{S}^1$ in total. 
Further taking into account the spatial symmetry, we only need to focus on 4 sinks and 5 saddles  for $(\theta_1,\theta_2) \in [0, \pi]\times [0,\pi]$, as shown in Table \ref{table1} and  Figure \ref{fig:phasespace2}.
In the figure,  the heteroclinic orbits between these fixed points are shown in arrowed lines.
The saddle point $sa_5$,  at the centre of the figure, is on the separatrix of all four sinks
in the phase space
and its unstable manifold has dimension 2.  All other four saddle points 
have one dimensional unstable manifold for each, i.e.  they are index-1 saddles.
\begin{center}
\begin{table}
\begin{tabular}{c|l|c||c|l|l}
\toprule
stable points & $\mathbb{S}^2\times\mathbb{S}^2$ & $(\theta_1,\theta_2)$
&
saddle points & $\mathbb{S}^2\times\mathbb{S}^2$ & $(\theta_1,\theta_2)$\\
\midrule
$si_1$& $ (\vec{e}_1, \vec{e}_1)$ & $ (0,0)$
&
 $sa_1$ & $(\vec{e}_2,\vec{e}_1)$ &  $(\pi/2, 0)$ 
 \\
$si_2$& $(\vec{e}_1,-\vec{e}_1)$ &  $(0,\pi)$
&
$ sa_2$ & $(\vec{e}_1,\vec{e}_2$)& $(0,\pi/2)$ \\
$si_3$& $(-\vec{e}_1,\vec{e}_1)$ & $(\pi, 0)$
&
$sa_3$  & $(\vec{e}_2,-\vec{e}_1)$ & $(\pi/2,\pi/2)$ 
\\
$si_4$ &$(-\vec{e}_1,-\vec{e}_1)$ & $(\pi,\pi)$
&
$ sa_4$ & $(-\vec{e}_1,\vec{e}_2)$ & $(\pi,\pi/2)$\\
&&&  
$ sa_5$ & $(\vec{e}_2,\vec{e}_2)$ & $(\pi/2,\pi/2)$\\
\toprule
\end{tabular}
\caption{\label{table1} 4 sinks and 5 saddles for $(\theta_{1},\theta_{2})\in [0,\pi]\times[0,\pi]$.}
\end{table}
\end{center}
 \begin{figure}[htbp]
 \begin{center}
  \includegraphics[width=0.46\textwidth]{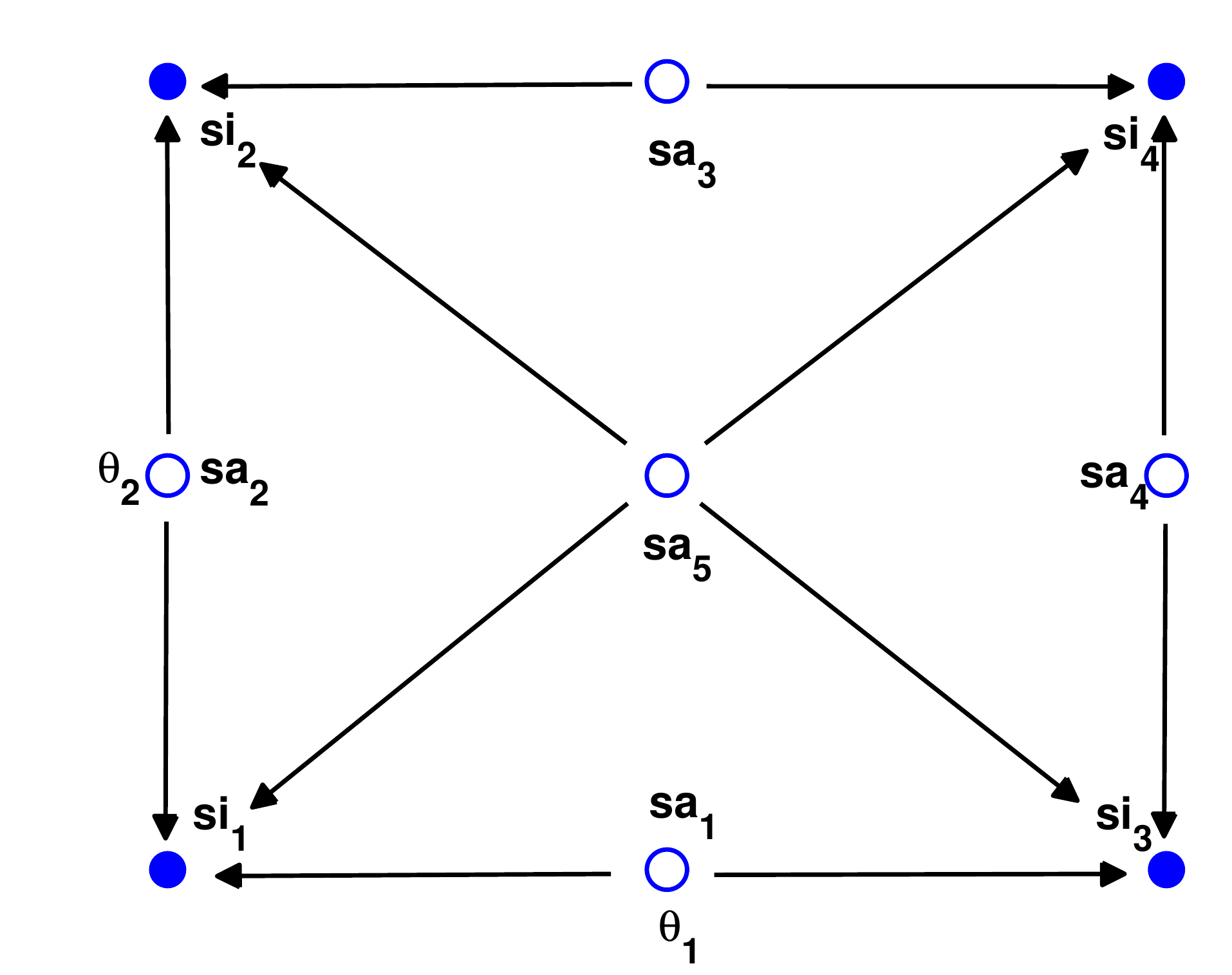}
  \caption{Fixed points in $\theta_1$-$\theta_2$ plane. Sinks are denoted by solid dots ($\bullet$) , saddles are denoted by circles ($\circ$).
  The arrows shows the heteroclinic orbits of the deterministic drift flow.  All saddles have index 1 except that $sa_5$ has index 2. }
  \label{fig:phasespace2}
  \end{center}
\end{figure}

The transition path we studied is   
from $si_2$ to $si_3$ which are two diagonal elements  in   Figure \ref{fig:phasespace2}. 
In solving minimization problem $\inf_{\phi}\hat{S}[\phi]$,  one  critical issue is how to locate the global solution rather than 
trapped by the local ones  \cite{KS-WZE2009}.  Since there is no efficient global minimization  solvers (we used   matlab subroutine  {\texttt{fmincon}}  for nonlinear optimization),
the selection of initial guess of  path is crucial.  We utilize the information of the heteroclinic orbits in Figure  \ref{fig:phasespace2}
and propose the following five routes as our initial guesses by choosing different transition states or intermediate states:
\begin{description}
\item[A] $si_2\to sa_5\to si_3$,
\item[B] $si_2\to sa_2\to sa_5\to si_3$,
\item[C] $si_2\to sa_2\to si_1\to sa_1\to si_3$,
\item[D] $si_2\to sa_3\to sa_5\to si_3$,
\item[E] $si_2\to sa_3\to si_4\to sa_4\to si_3$. 
\end{description}
Then, each choice of initial guess gives a local minimum action path
and the obtained minimized actions for the five solutions are plotted 
in Figure \ref{fig:action}.  The lowest value of these five curves gives the global minimum action.

\begin{figure}[htbp]
\begin{center}
  \includegraphics[width=0.6\textwidth]{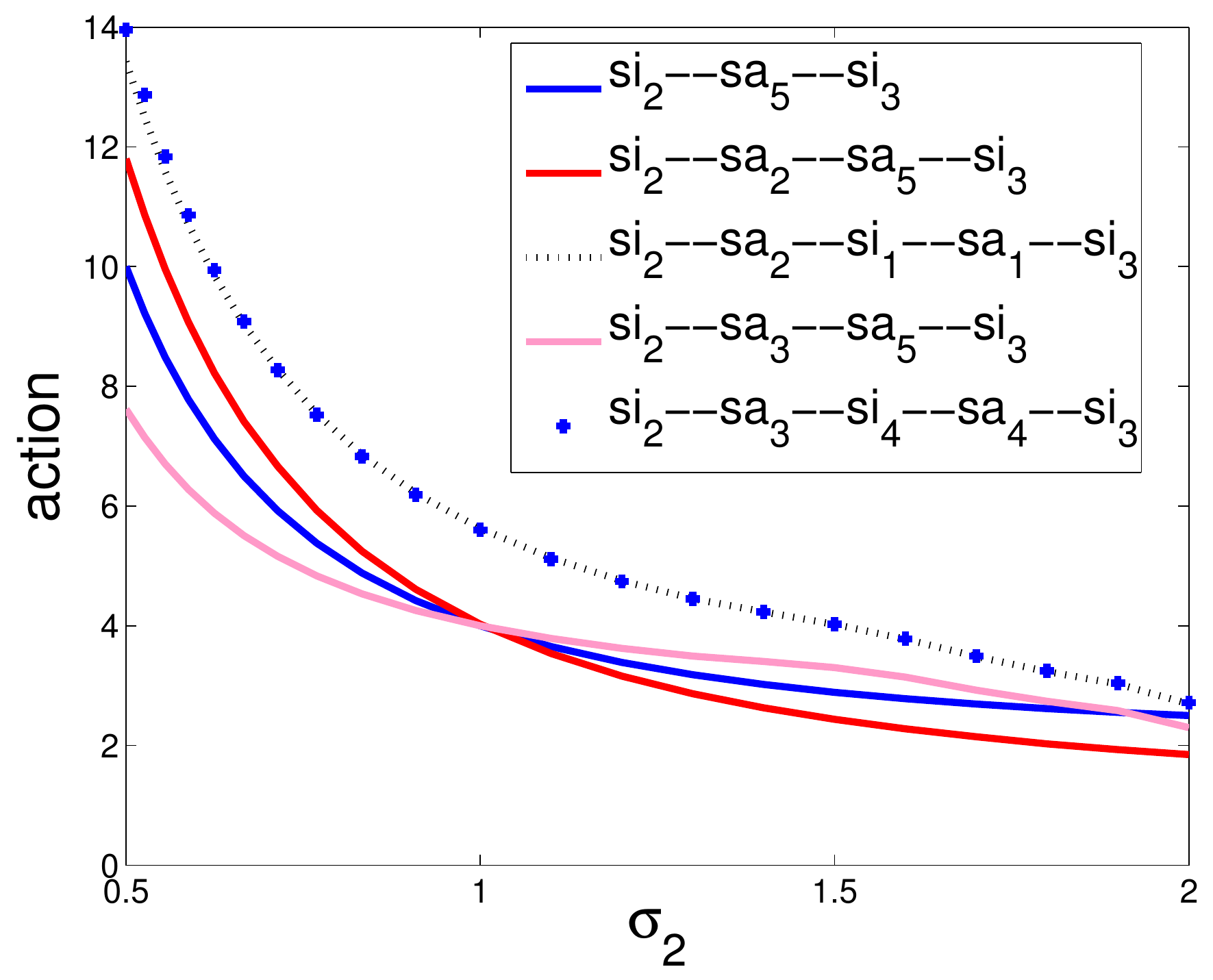}  
  \caption{Numerical values of the minimal actions corresponding to  five different initial guess paths by 
  varying   the diffusion coefficient $\sigma_2$.}
\label{fig:action}
\end{center}
\end{figure}

When $\sigma_1=\sigma_2$,  the same  global solution can be achieved from initial guesses {\bf A}, {\bf B} and {\bf D}.
 This global minimum action path is the diagonal line  ($si_2\to sa_5\to si_3$) in the  $\theta_1$-$\theta_2$ visualization (Figure \ref{fig:pathA}). 
However, when $\sigma_2 \neq \sigma_1$,   the symmetric path  ($si_2\to sa_5\to si_3$)  
is not the global minimal solution; in fact, the path for the global solution will
pass through index-1 saddle point $sa_2$ (if $\sigma_2>\sigma_1$) or $sa_3$ (if $\sigma_2<\sigma_1$).

Take $\sigma_2=1.2>\sigma_1=1$ as an example. The  transition  path corresponding to the global minimizer of the action functional 
is shown  in the right panel of  Figure  \ref{fig:pathB}.  The symmetry of the transition path is broken for this case of unequal diffusion coefficients. 
This  asymmetric path has three segments and accordingly  the transition process can be understood via three stages:  The first stage 
is from $si_2=(\vec{e}_1, -\vec{e}_1)$ to $sa_2=(\vec{e}_1, \vec{e}_2)$, 
where the first rod does not move much and only the second rod, which has the larger diffusion coefficient,  rotates 
in clockwise to the vertical position $\vec{e}_2$;  then, at the second stage which is from $sa_2$ to $sa_5=(\vec{e}_2,\vec{e}_2)$ , the second rod 
is almost still and ``waits" in the state $\vec{e}_2$ for the first rod to move from $\vec{e}_1$ to $\vec{e}_2$.
Once both rods reach the saddle state $sa_5$, the last state starts and both rods directly approach  the final state $si_3$ 
following the heteroclinic orbit in Figure \ref{fig:phasespace2} without   any aid from noise.

The above numerical results demonstrate  a selection mechanism:  the  rod with a  larger diffusion coefficient   $\sigma$ is 
subject to   large random perturbations with the same white noise realizations, and  thus it is easier to make transition movements  first.  We may call this rod as an ``active" rod.   After  this rod actively approaches a critical state ($\vec{e}_2$ here), it rests there, and 
 the interaction $U(\bx_1,\bx_2)$ starts to be the main contributor to influence the system  and the  previously still   rod (``passive" one) is attracted by
 $U$ from the active rod to the  critical state, from where the entire system has crossed all the barriers on the route of  the transition. 
What is unexpected  here is the ``sequentiality"  of the two rods' movement  during the first and the second transition stages. 
 In other words, ``sequentiality"   here means   the relative insensitivity to the other rod when  one rod is making progressive transition movement.
Taking an analogy of the so-called {\it reaction coordinate} in  chemical reactions,
we can think of $\theta_2$ as  an excellent candidate for reaction coordinate at the first  transition  stage
and   $\theta_1$  at the second stage. 
When  we varied $\sigma_2$ from $1$ to $2$ ($\sigma_1=1$ is fixed),  the numerical result shows the robustness of this set of reaction coordinates 
 especially 
at  the first stage from $si_2$ to $sa_2$.  Refer to  Figure \ref{fig:pathC} for the plots of 40  (global minimum action) paths for various values of $\sigma_2$
 by equally dividing the squared  $\sigma_2$  from $1$ to $4$.

 \begin{figure}[htbp]
 \centering
  \begin{subfigure}[b]{0.34\textwidth}
 {\includegraphics[width= \textwidth]{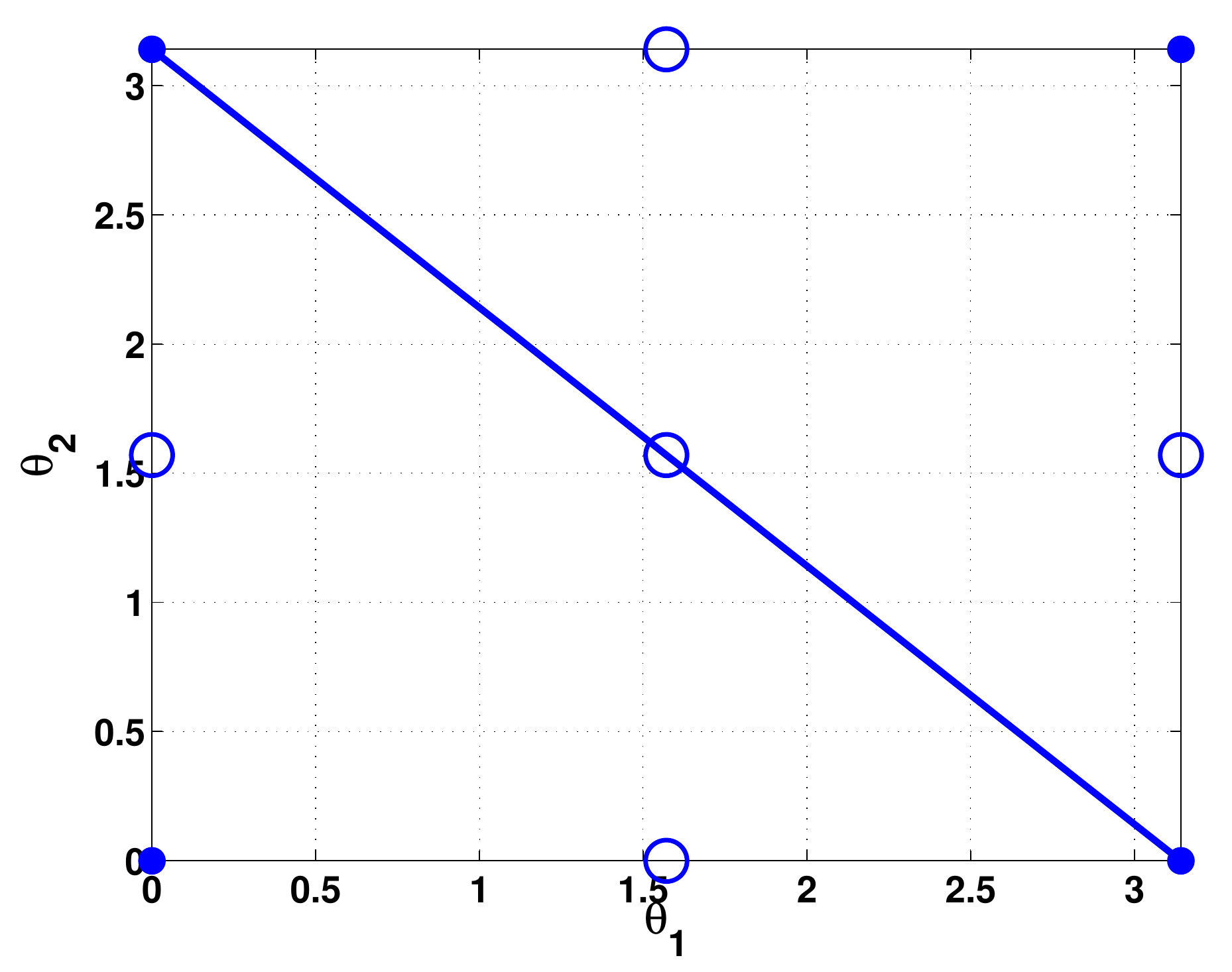}}  
              \caption{ $\sigma_2=1$}
\label{fig:pathA}        
\end{subfigure}%
 \begin{subfigure}[b]{0.34\textwidth}
 {\includegraphics[width= \textwidth]{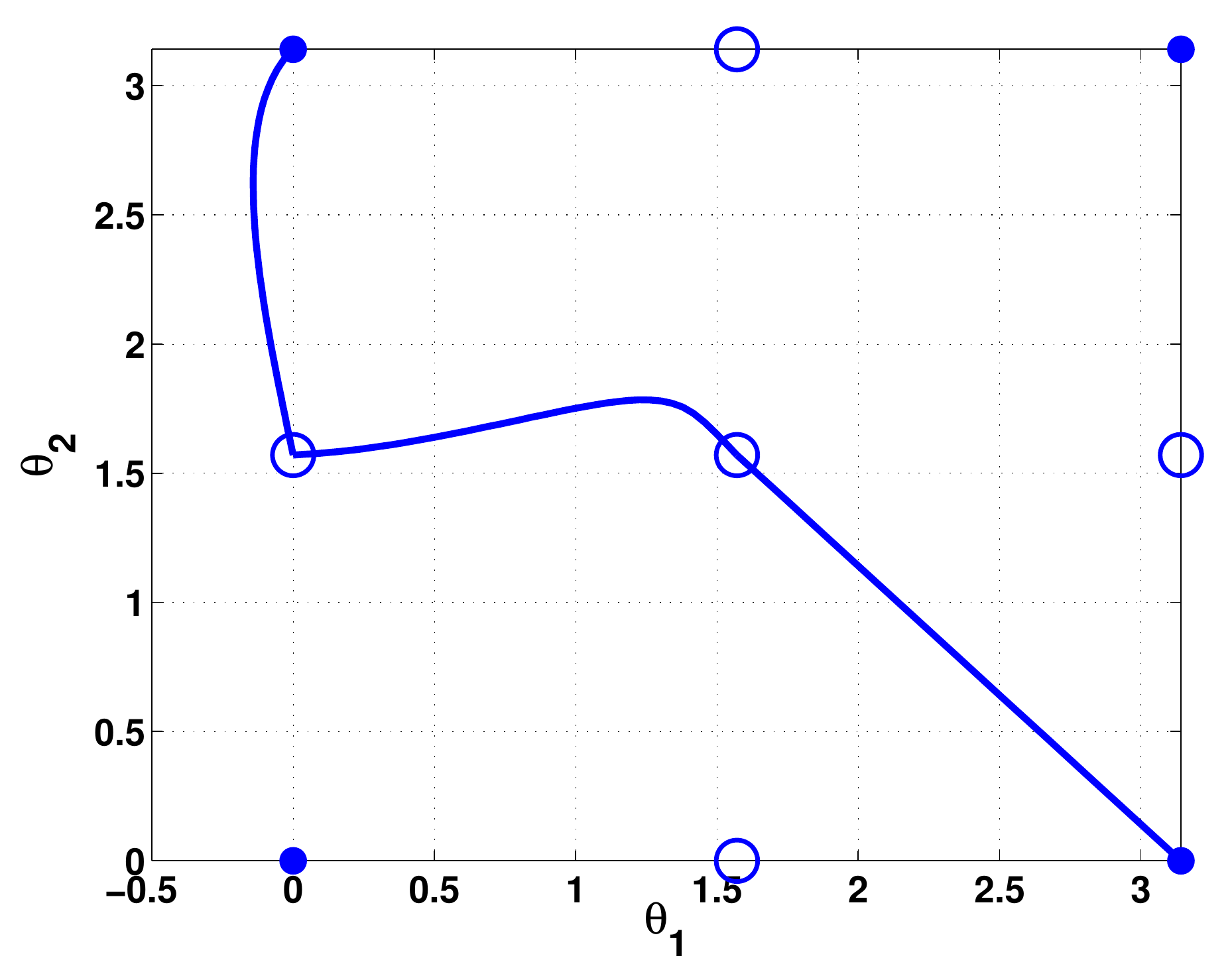}}  
              \caption{ $\sigma_2=1.2$}
\label{fig:pathB}        
\end{subfigure}%
\begin{subfigure}[b]{0.34\textwidth}
 {\includegraphics[width= \textwidth]{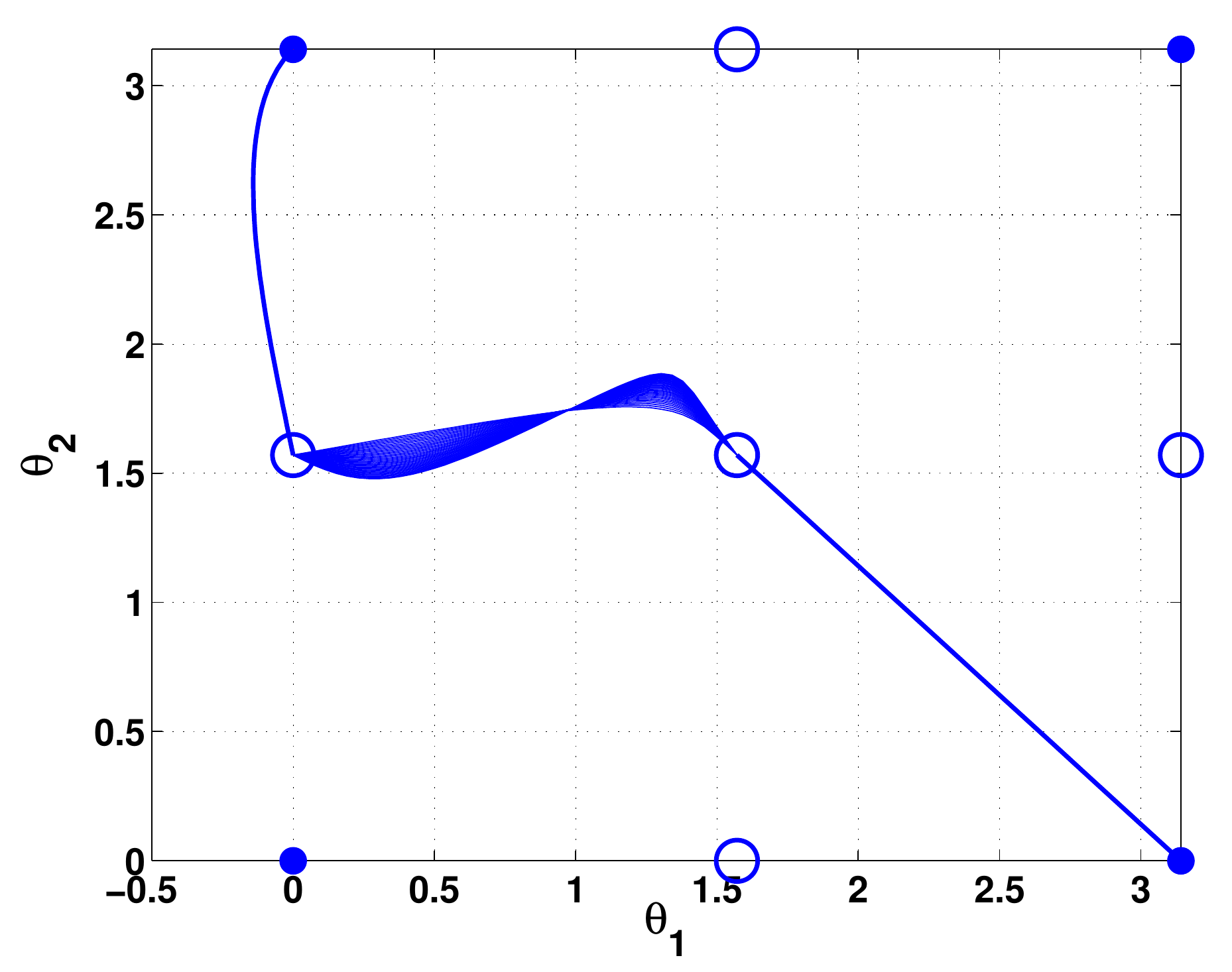}}  
              \caption{ $\sigma_2\in (1,2)$}
\label{fig:pathC}       
 \end{subfigure}%
  \caption{ Global minimum action path(s) from $si_2$ to $si_3$. }
\label{fig:path}
\end{figure}

In all, when the diffusion coefficients for the two rigid rods are identical,  the transition
path  is symmetric and both rods move simultaneously in the transitions. 
If one of the diffusion coefficients is adjusted,  then the rod molecule with larger diffusion amplitude 
will initiatively move into some intermediate state, then the other will follow in the similar fashion. 
The unbalance of the noise amplitudes  triggers an ordered process for each rod to make the transitions.

 \section{Outlook}
 
So far we have studied the most probable transition pathway for the stochastic dynamics of the type \eqref{eqn:SDE-Manifold}. In this setup, the drift and diffusion terms, in particular the projection $\Pi$, are explicitly known so that  the random motion does sit on the manifold $\mfd$. Another interesting but different setup is that the drift and diffusion terms are not known beforehand but left to be determined by the constraints. This situation is very common for problems in computational science. Let us illustrate this point with an example from the polymer science \cite{Doyle}.

\begin{figure}[htbp]
\begin{center}
\includegraphics[scale=0.36]{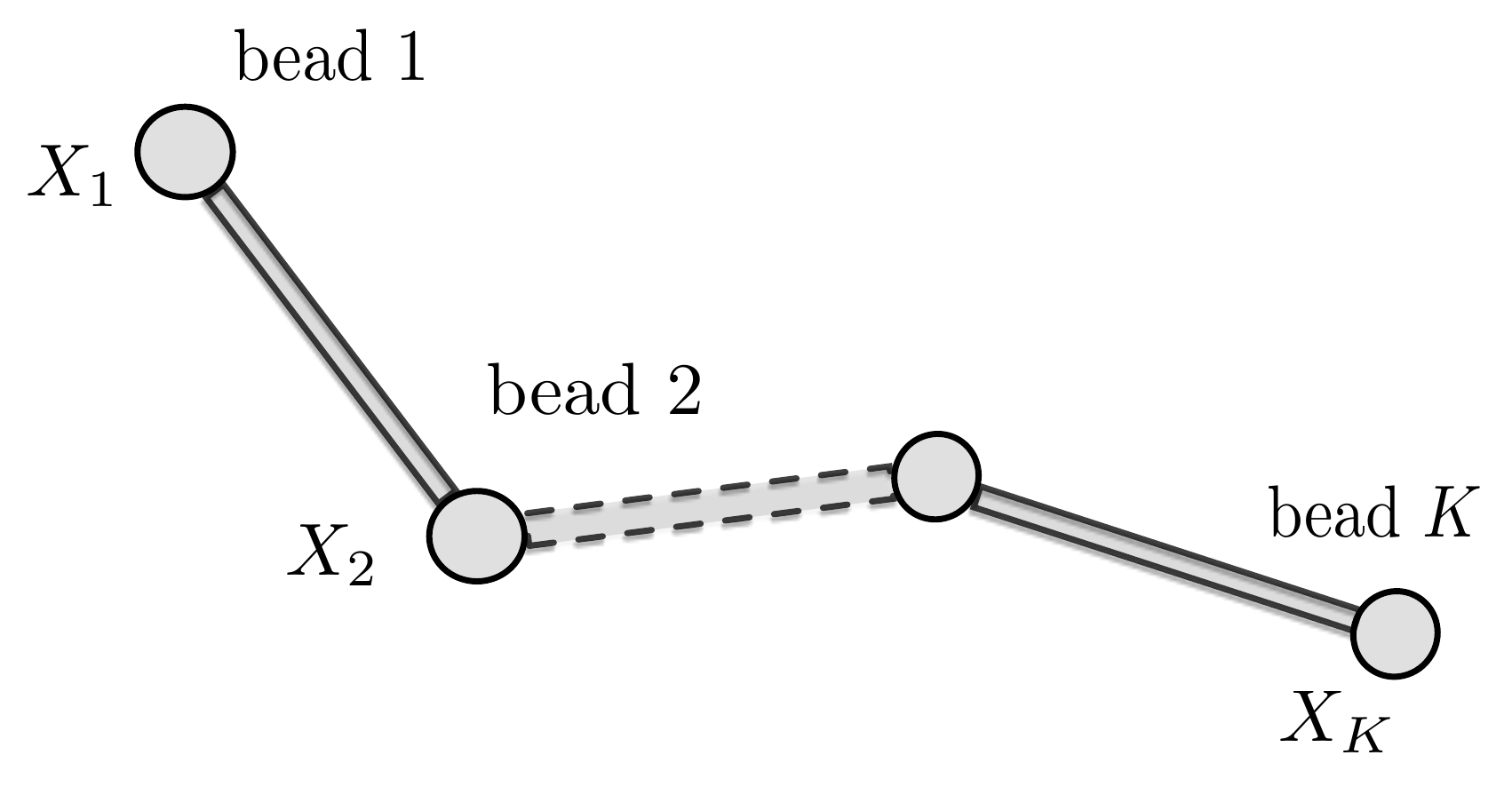}
\caption{The schematics of the bead-rod polymer chain with $K$ beads. }
\label{fig:beadrodchain}
\end{center}
\end{figure}

Consider a bead-rod polymer chain with $K$-beads (Figure \ref{fig:beadrodchain}), being described,  for instance,  by  the following  over-damped stochastic dynamics 
\begin{equation}\label{eq:bead-rod}
\d\bX_{i} = \Big(\ud(\bX_{i}) + (T_{i} \nd_{i} - T_{i-1} \nd_{i-1})\Big) \d t + \sqrt{\eps} \d\bW_{i}, \quad i=1,2,\ldots,K
\end{equation}
where $\bX_{i}$ is the coordinate of the $i$th bead, $\ud(\bX)$ is the fluid velocity at $\bX$,  
$\nd_{i} = (\bX_{i+1} - \bX_{i})/\|\bX_{i+1} - \bX_{i}\|$ for $i=1,2,\ldots,K-1$, and $T_{i}$ is the tension between the beads $i$ and $i+1$ such that the constraints
\begin{equation}\label{eq:bead-rod-constraint}
\|\bX_{i+1}-\bX_{i}\| = 1,\quad i=1,2,\ldots, K-1
\end{equation}
are satisfied. We take the convention that $T_{0}=0$ and $\nd_{K}=0$. It is obvious that the tension $\{T_{i}\}$ play the role of the Lagrange multipliers and it is only known after solving a nonlinear system. In this case, the previously considered formulation is not sufficient.

However, the transition pathway finding problem can also be formulated based on large deviation theory.  The straightforward application of the \FW theory to the system \eqref{eq:bead-rod} and \eqref{eq:bead-rod-constraint} gives the rate functional
\begin{equation}\label{eq:bead-rod-FWAction}
S_{T}[\phid] = \frac12\sum_{i=1}^{K}\int_{0}^{T}\|\dot\phid_{i}-\ud(\phid_{i}) - \bd_{i}\|^{2}\d t
\end{equation}
such that 
\begin{equation}\label{eq:bead-rod-FWActionConstraint}
\|\phid_{i+1}-\phid_{i}\| = 1,\quad i=1,2,\ldots, K-1,
\end{equation}
where $\bd_{i} = T_{i} \nd_{i} - T_{i-1} \nd_{i-1}$ and $\nd_{i} = (\phid_{i+1} - \phid_{i})/\|\phid_{i+1} - \phid_{i}\|$. Its geometric formulation can be obtained similarly as the derivations in \cite{Heymann2006}. We have 
$$
\hat{S} [\varphid] = \frac12\int_{0}^{1} \inpd{\varphid'}{ \hat{\pd}(\varphid,\varphid') }\d\alpha
$$
such that \eqref{eq:bead-rod-FWActionConstraint} is also satisfied. Here $(\hat{\pd}(\bx,\by),\lambda)$ is the unique solution of the system
$$H(\bx, \hat{\pd}) = 0,\quad H_{\pd}(\bx, \hat{\pd})=\lambda \by$$
where the Hamiltonian $H:\Real^{3K}\times \Real^{3K}\rightarrow \Real$ is defined as
$$H(\bx,\pd) =\frac12\inpd{\pd}{\pd} + \inpd{\ud(\bx)+\bd(\bx)}{\pd}.$$

Based on the obtained optimization problem with constraints or its relaxation form, we can compute the transition pathways correspondingly. We shall not develop
the study on this point here   since it is beyond the main goal of this paper. Further research on this topic will be a future study.

 \section{Summary}
  
 In this summary, we want to reiterate the mathematical  importance of  specifying  how the 
 constrained dynamical system is perturbed by noise  when one intends to  investigate the 
 transition paths in such constrained systems.  Here we considered  the SDE whose solutions
satisfy  constraints, i.e., stay on $\mfd$, for any $\eps$, rather than satisfy constraints  in the asymptotic sense.
 The asymptotic limit $\eps \downarrow 0$ is only applied in the large deviation result. 
In formulating the action functional, we take  the approach of using the local projection $\Pi$
to describe the constraints and solved the issue  of  degeneracy  brought by this projection operator in the augmented  Euclidean space $\Real^n$.  Certainly, it is possible to use the Lagrangian multiplier
to describe the constraints as mentioned in the previous section. 
 
 The constraints in the reversible case where the drift term is of gradient type and the diffusion coefficient  is isotropic constant
 actually do  not cause significant troubles in computations since the original string method still works by using the projected gradient force for  each image on string: It is essentially the same as one solves the deterministic gradient flow {\it on} the manifold.
 
 In the irreversible case, the calculation of constrained minimum action  need consider the generalized inverse of the projection operator
 unless the diffusion coefficient is  isotropic constant.   Additionally,  the resulting constrained optimization problem needs to be solved 
 by carefully choosing initial guesses to find the global solution.   
The initial guesses in our  example of  rigid rod models  are built on some prior understanding of  the phase spaces. 
In our study  that the liquid crystal molecules are  under the influence of   shear  or   possess  unequal diffusion constants,
the found global minimum action  pathways  reveal  very interesting non-equilibrium  phenomena,   and these phenomena are believed to be  generic  in irreversible systems
and deserve  further investigations.
 
 \section*{Acknowledgement}
T. Li acknowledge the support from NSFC under grants 11171009, 91130005 and the National Science Foundation for Excellent Young Scholars (Grant No. 11222114). X. Zhou acknowledges the  financial support of  CityU Start-Up Grant (7200301) and Hong Kong Early Career Schemes (109113).  
  
\bibliographystyle{amsalpha}

\bibliography{./ngre_copy_v3}

\end{document}